\documentclass[lettersize,journal]{IEEEtran}
\usepackage{amsmath,amsfonts}
\usepackage{array}
\usepackage[caption=false,font=footnotesize]{subfig}
\usepackage{textcomp}
\usepackage{stfloats}
\usepackage{url}
\usepackage{verbatim}
\usepackage{graphicx}
\usepackage{cite}
\usepackage[ruled,vlined,linesnumbered]{algorithm2e}
\SetKwRepeat{Repeat}{repeat}{until}  

\usepackage{threeparttable}
\usepackage{amsmath,graphicx,cite,amssymb}
\usepackage{amsfonts,multirow,bm,array,setspace,stfloats}
\usepackage{graphicx,float,cite,amssymb}
\usepackage{graphics}
 \DeclareGraphicsExtensions{.pdf,.jpeg,.png,.jpg}   
\usepackage{multirow,bm,bbm,array,setspace,dsfont}
\usepackage{textcomp}
\usepackage{psfrag}
\usepackage{color}
\usepackage{pstricks,enumerate}
\usepackage{bbm}
\usepackage{subfig}
\usepackage{makecell}
\usepackage{amsthm}

\newtheorem{proposition}{Proposition}

\newtheorem{theorem}{Theorem}

\newtheorem{remark}{Remark}

\newcommand{\bA}{\bm A}
\newcommand{\ba}{\bm a}
\newcommand{\bB}{\bm B}

\newcommand{\bI}{\bm{I}}

\newcommand{\bu}{\bm u}

\newcommand{\bbC}{\mathbb{C}}

\newcommand{\cN}{\mathcal{N}}
\newcommand{\cK}{\mathcal{K}}

\begin{document}
%
\title{Curved Beam Enabled Wireless Communications: Modeling, Analysis and Optimization}
\author{
	\IEEEauthorblockN{Jiawei Yao, \IEEEmembership{Student Member,~IEEE}, Xiaoren Xu, \IEEEmembership{Student Member,~IEEE}, Walid Saad, \IEEEmembership{Fellow,~IEEE},\\ Mingzhe Chen, \IEEEmembership{Senior Member,~IEEE} 
}\vspace{-1.3em} 
\thanks{
Jiawei Yao and Xiaoren Xu are with the Department of Electrical and Computer Engineering, University of Miami, Coral Gables, FL, 33146, USA (Email: \{jiaweiyao,xiaoren.xu\}@miami.edu).

Walid Saad is with the Bradley Department of Electrical and Computer Engineering, Virginia Tech, Alexandria, VA, 22305, USA. (Emails: walids@vt.edu).

Mingzhe Chen is with the Department of Electrical and Computer Engineering and Frost Institute for Data Science and Computing, University of Miami, Coral Gables, FL, 33146, USA (Email:  mingzhe.chen@miami.edu).

\textit{(Corresponding Author: Mingzhe Chen)}.
}
}


%


\maketitle

\begin{abstract}
In this paper, the problem of using curved beams to improve wireless communication performance in the presence of a blockage is studied. In particular, a transmitter equipped with a continuous aperture array can generate curved beams to serve multiple receivers by allowing signals to propagate along both straight and curved paths. To optimize the weighted sum-rate, a curved beam model is developed for controlling the beam steering, beam focusing, and beam curving functions, along with a segmented channel model to characterize practical channels induced by the blockage. Based on the introduced curved beam model, an optimization problem is posed with the goal of maximizing the weighted sum-rate of all users under a transmit power budget and physical constraints of curved beams. To solve this problem, the continuous aperture is first converted into finite summations via a discrete sampling of the continuous coordinate. Then, the performance gap between the ideal continuous aperture design and its practical discrete aperture approximation is analyzed. Based on the above discrete approximation, an iterative algorithm is developed to optimize curved beam control parameters. In particular, the original problem is reformulated as a trackable form via fractional programming (FP). Then, the transformed problem is solved by designing an enhanced block coordinate ascent (BCA) method which determines a surrogate-construction point leveraging the local descent from previous iterations, thereby accelerating convergence. Then, a proximal regularization term is included into the surrogate function to control the update magnitude and suppress aggressive update, thereby improving updates stability. Finally, the beam amplitudes are computed based on the effective channel gains. Simulation results show that the proposed method can improve the weighted sum-rate by up to $60\%$ compared to using only straight beam.

\end{abstract}

\section{Introduction}\label{sec:overview}

\IEEEPARstart{M}{ASSIVE} multiple-input multiple-output (MIMO) has become a key technology for next-generation wireless networks, as they exploit large antenna arrays and beamforming to enhance spectral efficiency and support massive connectivity \cite{6G_survey,blind_MIMO}. In conventional beamforming design, the transmitters adjust digital beamformings to generate straight beams and steer the beams toward target receivers. However, in non-line-of-sight (NLoS) environments (e.g., in the presence of blockages), signal strength of these straight beams will be attenuated due to the obstruction of the propagation path \cite{RIS_PAN,Kc,blind_multi,blind_twc}. This challenge becomes more serious in near-field millimeter-wave (mmWave) and terahertz (THz) bands \cite{mm_survey2,THzsurvey2}, due to high propagation loss and strong sensitivity to blockage. Meanwhile, near-field mmWave and THz networks provide finer-grained beam shaping. By enabling wave propagation along curved paths around blockages, curved beams offer a promising solution to signal attenuation suffered by straight beams in NLoS environments \cite{Cur-THz-conf,Cur-THz-NC,CurvingTHz_CE,CurZ,Airy-beam,AiryPRL}. However, designing curved beams for wireless networks faces some challenges including: i) how to model curved beam propagation in practical environments, ii) how to analyze relationship between curved beam parameters and system performance, as well as the performance gap between practical implementations and ideal designs, and iii) how to optimize curved beam parameters for improving communication performance.

\subsection{Related Works}
Recently, several works in \cite{Cur-THz-conf,Cur-THz-NC,CurvingTHz_CE,CurZ} studied the use of curved beams to improve wireless network performance. In particular, the authors in\cite{Cur-THz-conf} and \cite{Cur-THz-NC} studied a curved beam design in sub-THz wireless networks, where neural networks are used to extract the features of obstructed environments and optimize the curved beam related parameters. The work in \cite{Cur-THz-NC} further showed that by learning the parameters of curved beams, the transmitter can adapt the wavefront in real time under dynamic blockage conditions with much lower overhead than exhaustive search or iterative near-field simulations. The authors in \cite{CurvingTHz_CE} investigated the use of self-accelerating THz beams to generate curved wireless links around obstacles in near-field.
They developed a trajectory-engineering method for generating curved beams, and verified that such curved beams can outperform conventional straight beams. In \cite{CurZ}, the authors considered a near-field secure communication system based on curved caustic beams. They divided the transmit array into focusing and caustic sub-arrays, such that the focused sub-array concentrate the beam energy to the legitimate user while the caustic sub-array leads the beam around the eavesdropping region. However, most of these works \cite{Cur-THz-conf,Cur-THz-NC,CurvingTHz_CE,CurZ} focused on point-to-point or single-user scenarios. As such, the solutions in \cite{Cur-THz-conf,Cur-THz-NC,CurvingTHz_CE,CurZ} cannot be applied for multi-user curved beam communication scenarios since the curved beam parameters of different users are coupled through both desired signal enhancement and inter-user interference suppression. Meanwhile, the works in \cite{Cur-THz-NC,CurvingTHz_CE,CurZ} generated curved beams through heuristic and implementation-oriented methods (e.g., array-partitioning strategy) and hence, they do not provide a unified closed-form curved beam model that can support mathematical performance analysis, tractable optimization, and performance guarantees.


 To optimize the parameters of a curved beam, it is necessary to establish a tractable connection between the continuous-aperture beamforming model and the resulting communication performance metrics. Several works \cite{C-MIMO,Ouyang,Xidong,XD2,OY2} have investigated mathematically tractable beamforming optimization for continuous aperture arrays. The authors in \cite{C-MIMO} studied a continuous aperture enabled MIMO system. They developed a pattern-division multiplexing (PDM) framework by transforming the continuous pattern design into the optimization of projection coefficients on finite orthogonal bases. In \cite{Ouyang}, the authors considered continuous aperture-based MIMO systems and developed a weighted minimum mean square error-based beamforming design framework with calculus-of-variations-based closed form updates. The authors in \cite{Xidong} extended the continuous aperture optimization to multicast communications by developing a tractable block coordinate descent-based beamforming algorithm. In \cite{XD2}, the authors investigated a continuous aperture-based secure wireless communication system by jointly optimizing the information-bearing and artificial-noise source current patterns. The work in \cite{OY2} studied the optimal beamforming design for multi-user continuous aperture systems from a continuous-function optimization perspective, where a polyblock outer approximation algorithm was adopted. However, most of these works \cite{C-MIMO,Ouyang,Xidong,XD2,OY2} mainly focused on straight or generic beamforming optimization for continuous-aperture arrays under free-space channel models. Their closed-forms or efficiently solvable update rules are derived based on an explicit and analytical relationship between the continuous aperture beamforming and the resulting objective functions (e.g., received signal power, sum-rate). However, such analytical structures do not directly extend to curved beam design. This is because the curved beam parameters introduce additional nonlinear and coupled phase terms beyond the linear steering phase over the continuous aperture, while blockage further makes the end-to-end channel depend on segmented propagation in an inhomogeneous environment.

\subsection{Contributions}
The main contribution of this paper is a novel curved beam enabled wireless communication framework, that allows a transmitter to generate both straight and curved beams by adjusting curved beam parameters in blockage-prone environments. Our key contributions include:
\begin{itemize}
    \item We consider a wireless communication system in which a transmitter equipped with a continuous aperture array serves multiple receivers in the presence of a blockage. The transmitter can adjust the curved beam parameters that allow signals to propagate in both straight and curved path, thereby reducing attenuation induced by the blockage and improving data rates of receivers. To characterize the curved beam generation, we first develop a curved beam model that controls beam steering, beam focusing, and beam curving functions. Then, we present a segmented channel model to capture both the beam propagation process and the large-scale fading in blockage environments. Based on the curved beam model and channel model, we formulate an optimization problem whose goal is to maximize the weighted sum-rate of all users under the transmit power budget and physical constraints of curved beams.

    \item To solve this problem, one cannot directly optimize over the continuous aperture, since the original problem involves infinite-dimensional beamforming functions and continuous-domain integrations that may not admit closed-form expressions. To this end, we first discretize the continuous aperture into a finite number of sampled antennas and analyze the performance gap between the ideal continuous aperture design and its practical discrete aperture approximation. In particular, we derive an upper bound on the weighted sum-rate gap, which reveals how the aperture discretization and curved beam related parameters affect the achievable communication performance. The analytical result also provides a theoretical insight for converting the original continuous aperture optimization problem into a tractable discrete aperture optimization.

    \item Based on the above discretization, we develop an iterative algorithm to optimize the curved beam parameters. In particular, we reformulate the original problem via fractional programming (FP) to transform the non-convex sum-rate objective function as a trackable linear and quadratic function. Then, we solve the transformed problem via designing an enhanced block coordinate ascent (BCA) method which determines a surrogate-construction point leveraging the local descent from previous iterations, thereby accelerating convergence. Then, a proximal regularization term is included into the surrogate function to control the update magnitude and suppress aggressive update, thereby improving updates stability. Finally, we calculate the beam amplitudes according to the effective channel gains.
\end{itemize}
Simulation results show that the proposed method can improve the weighted sum-rate by up to $60\%$ and $12\%$ compared with using only straight beam methods and other curved-beam baseline schemes. \textit{To the best of our knowledge, this is one of the first works to jointly study curved beam modeling, performance analysis, and parameter optimization for wireless communications in presence of a blockage region.}

The rest of the paper is organized as follows. The curved beam model and problem formulation are described in Section~\ref{sec:model}. In Section~\ref{sec:analysis}, we analyze the performance gap between the continuous aperture array and its discrete counterpart. The proposed optimization approach for curved beam parameters, along with the implementation and complexity analysis, is studied in Section~\ref{sec:solution}. Simulation results are presented and analyzed in Section~\ref{sec:experiments}, and conclusions are drawn in Section~\ref{sec:conclusion}.  

The notation hereinafter and throughout is summarized as follows. We use $a$, $\ba$ and $\bA$ to denote a scalar, a vector and a matrix, respectively. For a matrix $\bA$, $\bA^*$, $\bA^{-1}$, $\bA^{\dag}$, $\bA^\top$, and $\bA^\mathsf{H}$ denote the conjugate, inverse, pseudo-inverse, transpose, and Hermitian of matrix $\bA$, respectively. $\| \bA\|_F$ denotes the Frobenius norm. We use $\{\ba\}$ and $\{\bA\}$ to denote the collections of vectors $\ba$ and $\bA$. We use $\text{diag}(\ba)$ to denote the diagonal matrix with entries of $\ba$ on the diagonal. $N$-dimensional real vectors and complex vectors are denoted by $\mathbb{R}^{N\times 1}$ and $\bbC^{N\times 1}$. $N\times M$-dimensional complex matrices is denoted by $\bbC^{N\times M}$, respectively. We use $\mathcal{CN}(\cdot,\cdot)$ to denote a complex Gaussian distribution; $\bI_n$ to denote a $n\times n$ identity matrix. We use $\bA\odot \bB$ to denote the Hadamard product of matrix $\bA$ and $\bB$.


\section{System Model and Problem Formulation}\label{sec:model}

Consider a near-field wireless communication system in which a transmitter equipped with a continuous aperture array serves a set $\cK$ of $K$ receivers equipped with a uni-polarized antenna \cite{Ouyang,Cur-THz-conf,C-MIMO}. In this scenario, a blockage $\bm\Omega_\mathrm{o}$ exists such that the transmission link between the transmitter and the receivers may be NLoS. To reduce the signal attenuation caused by the blockage, the transmitter uses curved beams for signal transmission by adjusting the beam parameters (e.g., beam steering angle and beam curving ratio). Using curved beams allows the signals to propagate in a curved path instead of a straight line, thus avoiding signal transmission through the blockage, as shown in Fig.~\ref{fig:system}. Our overarching goal is to optimize curved beam parameters so as to maximize the data rate of receivers under a scenario with blockage. Next, we first introduce the curved beam model. After that, we present the models of signals received by the receivers and the achievable rate of each receiver. Finally, we formulate the optimization problem.

\subsection{Curved Beam Modeling}
We consider a two-dimensional (2D) Cartesian coordinate system whose transmit array is centered at the origin  $(0,0)$. An arbitrary point on the transmit array is given by $\bm \ell=(0,\ell)$ with $\ell\in\mathcal{L}=[-L/2,L/2]$, where $L$ is the linear dimension of the array. 

To model a curved beam for receiver $k\in\cK$, we need to define three functions: 1) beam steering function that steers the beam toward a desired direction, 2) beam focusing function that determines the start
point of beam curving, and 3) beam curving function that determines how the beam is curved from the start point. In particular, the beam steering function is given by \cite{Cur-THz-conf,Cur-THz-NC,VanTrees2002}:
\begin{align}\label{eq:beam_steering}
    c_1(\theta_k)= -\frac{2\pi \sin\theta_k}{\lambda},
\end{align}
where $\theta_k$ is the steering angle  associated with receiver $k$. As shown in Fig.~\ref{fig:system}(a), 
$\theta_k$ determines the beam angle with respect to the $+x$-axis.

The beam focusing function determines the region where signals are concentrated and also decides the start point of beam curving. This is given by~\cite{Cur-THz-conf,Cur-THz-NC,Goodman2005}:
\begin{align}\label{eq:beam_focusing}
c_2(f_k)=-\frac{\pi}{\lambda f_k},
\end{align}
where $f_k$ is the distance from the transmit array center $(0,0)$ to the signal concentrated point. As shown in Fig.~\ref{fig:system}(b), the beam is achieved by combining \eqref{eq:beam_steering} and \eqref{eq:beam_focusing}. Here, a small value of $f_k$ indicates that the beam will be curved near to the transmit array.

The beam curving function determines how the beam is curved from the signal concentrated point defined in \eqref{eq:beam_focusing}, which is represented as~\cite{AiryPRL,Cur-THz-conf,Cur-THz-NC,Airy-beam}:
\begin{align}\label{eq:beam_curving}
    c_3(B_k)=\frac{(2\pi B_k )^3}{3} ,
\end{align}
where $B_k$ is the curvature of the curved beam. As shown in Fig.~\ref{fig:system}(c), the curved beam is generated by the equations from \eqref{eq:beam_steering} to \eqref{eq:beam_curving}. Here, a smaller $B_k$ results in a flatter beam.

Given \eqref{eq:beam_steering}--\eqref{eq:beam_curving}, the curved beamforming function $J_k(\ell;\theta_k,f_k,B_k)\in\bbC$ of receiver $k$ at point $(0,\ell)$ of the continuous aperture will be:
\begin{align}\label{eq:get_J_k}
    J_k(\ell;\theta_k,f_k,B_k)=\alpha_k( \ell)e^{j\beta_k(\ell;\theta_k,f_k,B_k)},
\end{align}
where $\alpha_k(\ell)\ge 0$ is the amplitude that controls the beam strength, and
\begin{align}\label{eq:beamforming_phase}
    \beta_k(\ell;\theta_k,f_k,B_k)=c_1(\theta_k)\ell+c_2(f_k)\ell^2+c_3(B_k)\ell^3,
\end{align}
is the phase of the curved beamforming function. From \eqref{eq:beamforming_phase}, we see that unlike standard straight beams \cite{Cur-THz-conf,Cur-THz-NC}, which depend only on the beam steering function $c_1(\theta_k)$, a curved beam depends on beam steering, beam focusing, and beam curving functions with parameters $\{ \theta_k,f_k, B_k \}$.

\subsection{Received Signal and Achievable Rate} 
\begin{figure*}[t]
\begin{align}\label{eq:inter_planes}
    \mathcal{M}=\Big\{\bm r_m=(x_m,y_m) \Big|\, \bm r_0=\bm \ell, \ \bm r_M=\bm p_k, \  x_m=m\times \Delta_x, \ \Delta_x=x_k/M, \ y_m\in\mathbb{R},\ m=1,\ldots,M-1\Big\}.
\end{align}
\hrule 
\end{figure*}
Since the diffraction behavior of electromagnetic waves is altered in the presence of blockages, directly characterizing the end-to-end channel from the transmit aperture to receiver $k\in\cK$ is intractable. To address this issue, we discretize the signal propagation region from the transmit aperture to receiver $k\in\cK$ into $M$ sub-regions, such that the signal propagation in each sub-region can be characterized by the free-space channel model \cite{Cur-THz-conf,C-MIMO}. In particular, let $\bm r_m=(x_m,y_m)$ be the intermediate point in sub-region $m$, and let $\mathcal M$ in \eqref{eq:inter_planes} represent the set of all intermediate points along the propagation path.

 Let $\bm p_k=(x_k,y_k)$ be the 2D positions of the receiver $k\in\cK$ in the free-space. The effective channel of a curved beam from an arbitrary point $\bm \ell$ on the transmit aperture to receiver $k$ is
\begin{align}
    \tilde{h}(\bm p_k,\bm \ell)=\sqrt{\epsilon_k}\times\frac{\hat{h}(\bm p_k,\bm \ell)}{\sqrt{\int_{\mathcal{L}}\big|\hat{h}(\bm p_k,\bm \ell)\big|^2 d\ell}},
\end{align}
where $\beta_k$ is the large-scale pathloss coefficient,
\begin{align}\label{eq:eff_channel}
\hat{h}(\bm p_k,\bm \ell)
=
\int_{\mathbb{R}}\cdots\int_{\mathbb{R}}
\bigg[\prod_{m=1}^{M} \Gamma(&\bm r_m)h(\bm r_m,\bm r_{m-1})\bigg]\nonumber\\
&\quad \times dy_1\cdots dy_{M-1}.
\end{align} is the raw effective channel \cite{Cur-THz-conf,Cur-THz-NC}, which characterizes the propagation geometry and blockage effect along the transmission path. In \eqref{eq:eff_channel},
\begin{align}
    h(\bm{r}_m,\bm r_{m-1})=\bigg(&\frac{1}{\|\bm{r}_m-\bm r_{m-1} \|}-\frac{j2\pi}{\lambda}\bigg)\nonumber\\&\times\frac{|x_m-x_{m-1}|}{2\pi \|\bm{r}_m-\bm r_{m-1} \|^2}e^{j\frac{2\pi}{\lambda}\|\bm{r}_m-\bm r_{m-1} \|}.\label{eq:free-space channel}
\end{align}
is the free-space channel between two points $\bm r_m$ and $\bm r_{m-1}$ located in two adjacent sub-regions (i.e., regions $m$ and $m-1$) \cite{Ouyang,Cur-THz-conf,C-MIMO,Cur-THz-NC,Goodman2005},
and
\begin{align}\label{eq:blockage_function}
\Gamma(\bm r_m)=
\begin{cases}
\xi_m, & \bm r_m\in\bm\Omega_\mathrm{o},\\
1, & \bm r_m\notin\bm\Omega_\mathrm{o},
\end{cases}
\qquad m=1,\ldots,M,
\end{align}
is a blockage indicator with $\Gamma(\bm r_m)=\xi_m\in\bbC,\, 0<|\xi_m|<1$ representing that the signal transmission at $\bm r_m$ is attenuated by blockage, while $\Gamma(\bm r_m)=1$ indicating that no blockage is present at $\bm r_m$.

\begin{figure*}[t]
\centering
\includegraphics[width=16cm]{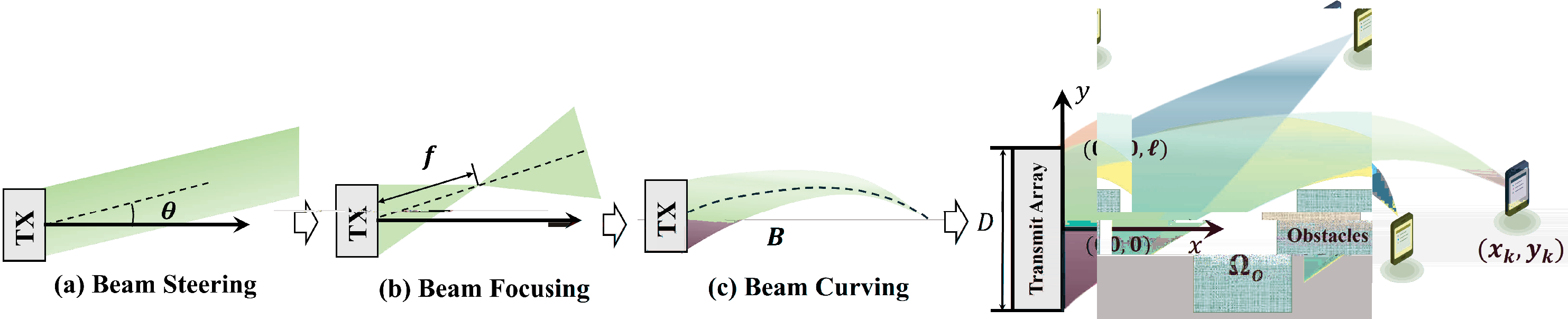}
\caption{ Curved beams enabled wireless communications.}
\label{fig:system}
\vspace{-1em}
\end{figure*}

Based on \eqref{eq:get_J_k} and \eqref{eq:eff_channel}, with the transmitted symbol $s_k\sim\mathcal{CN}(0,1)$, the received signal at receiver $k\in\cK$ is \cite{Ouyang,C-MIMO,Cur-THz-conf,Cur-THz-NC}
\begin{align}\label{eq:received_signal}
y_k
=
&\underbrace{\int_{\mathcal{L}}J_k(\ell;\theta_k,f_k,B_k)\, \tilde h(\bm{p}_k,\bm \ell)\,s_k\,d\ell}_{\text{desired signal}}
\nonumber\\& +
\underbrace{\sum_{i=1,\,i\neq k}^{K}\int_{\mathcal{L}} J_i(\ell;\theta_i,f_i,B_i)\, \tilde h(\bm{p}_k,\bm \ell)\,s_i\,d\ell + n_k}_{\text{inter-receiver interference}}.
\end{align}
where $n_k\sim\mathcal{CN}(0,\sigma_k^2)$ is an independent white Gaussian variables. Hence, the signal-to-interference-plus-noise ratio (SINR) at receiver $k\in\cK$ is
\begin{align}\label{eq:SINR}
 &\gamma_k\big(\{\theta_i,B_i,f_i\}_{i\in\cK}\big)\nonumber\\=\ \ &\frac{\Big|\int_{\mathcal{L}} J_k(\ell;\theta_k,f_k,B_k)\tilde h(\bm{p}_k,\bm \ell)d\ell\Big|^2}{\sum_{i=1, i\neq k}^K \Big|\int_{\mathcal{L}} J_i(\ell;\theta_i,f_i,B_i) \tilde h(\bm{p}_k,\bm \ell) d\ell\Big|^2+\sigma_k^2}.
\end{align}
The resulting achievable communication rate of receiver $k\in\cK$ is
\begin{align}
    R_k\big(\{\theta_i,B_i,f_i\}_{i\in\cK}\big)=\log_2\big(1+\gamma_k\big(\{\theta_i,B_i,f_i\}_{i\in\cK}\big)\big).
\end{align}

\subsection{Problem Formulation}\label{subsec:problem}
Given the considered curved beam system, we aim to maximize the weighted sum-rate of all receivers to ensure robust near-field communication within a network with blockages. The optimization variables include the steering angle $\theta_k$, the focusing distance $f_k$, and the curvature coefficient $B_k$. The resulting optimization problem is
\begin{subequations}
\label{pro:snr_max}
\begin{align}
    &\underset{\{\theta_i,f_i,B_i\}_{i=1}^K}{\max}  \ \ \ \sum_{k=1}^K \omega_k \, R_k\big(\{\theta_i,B_i,f_i\}_{i\in\cK}\big) \tag{\ref{pro:snr_max}}\\
    &\quad \ \ \,    \text {s.t.}  \  \qquad \, \,  \sum_{k=1}^K\int_{\mathcal{L}}\Big |J_k(\ell;\theta_k,f_k,B_k)\Big|^2\,d\ell \leq P_\mathrm{T}, \label{cont:power}\\ &\qquad\qquad \qquad -\frac{\pi}{2}\le \theta_k\le \frac{\pi}{2},\quad\ \ \forall k\in\mathcal K,\label{cont:theta}\\ &\qquad\qquad \qquad 0\leq f_k\le x_k,\qquad\ \  \, \forall k\in\mathcal K,\label{cont:f}\\ &\qquad\qquad \qquad -B_k^{\max}\leq B_k\leq B_k^{\max},\quad\ \forall k\in\mathcal K. \label{cont:B}
\end{align}
\end{subequations} 
where $\omega_k$ is a weight parameter of receiver $k$, used to reflect or balance receiver priority; $P_\mathrm{T}$ in constraint \eqref{cont:power} is the power budget of the transmitter \cite{C-MIMO}, \eqref{cont:theta} and \eqref{cont:f} ensures that the signal is steering and concentrated along the propagation direction, and \eqref{cont:B} prevents excessive beam curvature and remains within a practically realizable range.

The problem in \eqref{pro:snr_max} is challenging to solve due to the following reasons. First, the objective function is nonconvex with respect to the optimization variables, since it involves SINRs (i.e., $\gamma_k\big(\{\theta_i,B_i,f_i\}_{i\in\cK}\big)$) that are fractional and nonlinear of the variables $\{ \theta_i,f_i,B_i \}_{i=1}^K$. Second, the phase profile $e^{j\beta_k(\ell;\theta_k,f_k,B_k)},k\in\cK$ is periodic, which makes the objective admits multiple local optima. Third, the effective channel in \eqref{eq:eff_channel} is generated via iterative propagation across multiple intermediate regions, which results in a complicated and high-dimensional mapping from $\{ \theta_k,f_k,B_k \}_{k=1}^K$ to the objective function. Finally, both the objective function (i.e., achievable rate $R_k$) and constraints \eqref{cont:power} have integrations over the continuous aperture coordinate $\ell\in[-L/2,L/2]$. These integrals do not have closed-form expressions and make the gradient computation difficult.

\section{Performance Gap Between Continuous and Discrete Aperture}\label{sec:analysis}
To solve the problem in \eqref{pro:snr_max}, we first change integrations of objective functions in \eqref{pro:snr_max} and the constraint in \eqref{cont:power} 
into finite summations via discretely sampling the continuous aperture coordinate $\ell\in\mathcal{L}$. However, such approximation may introduce discretization errors.
Therefore, in this section, we analyze the approximation gap in terms of the weighted sum-rate.

\subsection{Performance Gap on  Weighted sum-rate}
We first assume that the continuous transmit aperture is discretized into $N$ uniformly spaced elements with sampling interval \(\Delta_\ell=L/N\le \lambda/2\), such that the transmitter is equipped with a set $\mathcal{N}$ of $N$ discrete antennas, whose positions are given by $\bm \ell_n=(0,\ell_n), n\in\cN$, where  
\begin{align}
    \ell_n=-\frac{L}{2}+\Big(n-\frac{1}{2}\Big)\times\Delta_\ell.
\end{align}
Then, the discrete beamforming vector for receiver $k\in\mathcal{K}$ is
\begin{align}\label{eq:dis_beamforming}
    \bm w_k(\theta_k,f_k,B_k)
    =
    \bm{\alpha}_k\odot e^{j\bm{\beta}_k(\theta_k,f_k,B_k)}
    \in\mathbb{C}^{N\times1},
\end{align}
where $\bm{\alpha}_k\triangleq[\alpha_k(\ell_1),\ldots,\alpha_k(\ell_N)]^\top,$ and
\begin{align}
    \bm{\beta}_k(\theta_k,f_k,B_k)
    \triangleq
    \Big[
    &c_1(\theta_k)\ell_1+c_2(f_k)\ell_1^2+c_3(B_k)\ell_1^3,
    \ldots,
    \nonumber\\
    &
    c_1(\theta_k)\ell_N+c_2(f_k)\ell_N^2+c_3(B_k)\ell_N^3
    \Big]^\top.
\end{align}
Based on \eqref{eq:eff_channel} and \eqref{eq:dis_beamforming}, with the transmitted symbol $s_k\sim\mathcal{CN}(0,1)$, the received signal at receiver $k\in\mathcal{K}$ is 
\begin{align}\label{eq:dis_y}
    \tilde{y}_k=\sum_{i=1}^K \Delta_\ell\tilde{\bm h}_k^\mathsf{H}\bm w_i(\theta_k,f_k,B_k)s_i+n_k,
\end{align}
where $\tilde{\bm h}_k=\big[\tilde{h}^*(\bm p_k,\bm \ell_1),\ldots,\tilde{h}^*(\bm p_k,\bm \ell_N)\big]\in\mathbb{C}^{N\times 1}$. Hence, the SINR at receiver $k\in\mathcal{K}$ is
\begin{align}\label{eq:dis_rate}
    \tilde{\gamma}_k\big(\{\theta_i,B_i,f_i\}_{i\in\cK}\big)=\frac{\Delta_\ell^2|\tilde{\bm h}_k^\mathsf{H}\bm w_k(\theta_k,f_k,B_k)|^2}{\sum_{i\neq k}\Delta_\ell^2| \tilde{\bm h}_k^\mathsf{H}\bm w_i(\theta_i,f_i,B_i)|+\sigma_k^2}.
\end{align}
The resulting achievable communication rate of receiver $k\in\mathcal{K}$ is 
$\tilde{R}_k\big(\{\theta_i,B_i,f_i\}_{i\in\cK}\big)=\log_2\big(1+\tilde\gamma_k\big(\{\theta_i,B_i,f_i\}_{i\in\cK}\big)\big)$.

Next, we analytically characterize the weighted sum-rate gap incurred by discretizing the continuous aperture array into $N$ discrete antennas.

\begin{theorem}\label{theorem:gap}
    Assume that the effective channel $\tilde{h}(\bm p_0,\bm \ell)$ in \eqref{eq:eff_channel} is twice differentiable with respect to $\ell$, and all aperture amplitudes are identical, i.e., $\alpha_k(\ell)=\alpha\geq0$, $\forall k\in\cK,\ \forall \ell\in\mathcal L$, the weighted sum-rate gap between the continuous aperture array and the discrete array formed by sampling the same aperture with $N$ elements is
    \begin{align}\label{eq:theorem_1_gap}
    &\bigg|\sum_{k=1}^K \omega_k R_k- \sum_{k=1}^K \omega_k \tilde R_k   \bigg|\nonumber\\\leq\, &\sum_{k=1}^K\bigg[\frac{\omega_k}{N^2\ln{2}}\bigg(\frac{2|g_{k,k}|}{\sigma_k^2}D_{k,k}+\frac{2|g_{k,k}|^2}{\sigma_k^4}\sum_{i\neq k}|g_{k,i}|D_{k,i}\bigg)\nonumber\\
    &\qquad\qquad\ +\frac{\omega_k}{N^4\ln{2}}\bigg(\frac{D_{k,k}^2}{\sigma_k^2}+\frac{|g_{k,k}|^2}{\sigma_k^4}\sum_{i\neq k}D^2_{k,i}\bigg)\bigg],
\end{align}
where 
\begin{align}
    g_{k,i}&=\int_{\mathcal{L}} J_{i}(\ell)\tilde h(\bm p_k,\bm \ell)\,d\ell,\\
    D_{k,i}&= \frac{L^3}{24}\times\Big(\tilde\Psi_{i}^{(2)}\tilde\Sigma_k^{(0)}+\tilde\Psi_{i}^{(0)}\tilde\Sigma_k^{(2)}+2\tilde\Psi_{i}^{(1)}\tilde\Sigma_k^{(1)}\Big),\label{eq:Cki}
\end{align}
and
\begin{align}
    &\tilde\Psi_i^{(0)}=\alpha,\\
        &\tilde\Psi_i^{(1)}= \alpha\bigg(\frac{2\pi}{\lambda}|\sin\theta_i|+\frac{2\pi L}{\lambda f_i}+\frac{(2\pi B_i)^3L^2}{4}\bigg),\label{eq:T_1}   
\end{align}
\begin{align}
    &\tilde\Psi_i^{(2)}= \alpha
    \Bigg[
    \frac{4\pi}{\lambda f_i}
    +(2\pi B_i)^3L+
    \bigg(
    \frac{2\pi}{\lambda}|\sin\theta_i|\nonumber\\
    &\qquad\qquad\qquad\qquad\qquad
    +\frac{2\pi L}{\lambda f_i}
    +\frac{(2\pi B_i)^3L^2}{4}
    \bigg)^2
    \Bigg],\label{eq:T_2}
\end{align}
\begin{align}
    &\tilde{\Sigma}_k^{(q)}= \sqrt{\frac{{\epsilon_k}}{\int_{\mathcal{L}}\big|\hat{h}(\bm p_k,\bm \ell)\big|^2 d\ell}}\bigg(\frac{\pi}{\Delta_x}\bigg)^Mc_q c_0^{M-1},\ q=0,1,2,\\
    &c_0=\frac{1}{2\pi}+\frac{\Delta_x}{\lambda},\\
            &c_1=\frac{\Delta_x}{2\pi}\times\bigg(\frac{4\pi^2}{\lambda^2}+\frac{6\pi}{\lambda\Delta_x}+\frac{3}{\Delta_x^2}\bigg),\\
            &c_2=\frac{\Delta_x}{2\pi}\times\bigg(\frac{8\pi^3}{\lambda^3}+\frac{24\pi^2}{\lambda^2\Delta_x}+\frac{30\pi}{\lambda\Delta_x^2}+\frac{15}{\Delta_x^3}\bigg).
\end{align}
\end{theorem}
\begin{proof}
     The proof is in the Appendix.
\end{proof}

From \textbf{Theorem~\ref{theorem:gap}}, we can make several key observations, as discussed next. 


\begin{remark}[Scaling with the aperture discretization $N$]\label{remark1}
    From \eqref{eq:theorem_1_gap}, we see that the weighted sum-rate gap depends on $1/N^2$, which is consistent with classical discretization results for continuous aperture integrations in  \cite{C-MIMO}. When $N$ is large enough (e.g., $N\geq\sqrt{C/\varepsilon}$ where $C$ represents other coefficients in \eqref{eq:theorem_1_gap} and $\varepsilon$ is tolerance), the continuous aperture beamforming function $J_k(\ell),\forall \ell\in\mathcal{L},k\in\cK$ in \eqref{eq:get_J_k} can be represented by the corresponding $N$--element discrete vector $\bm w_k\in\mathbb{C}^{N\times1}$ in \eqref{eq:dis_beamforming}, without incurring a significant performance loss in weighted sum-rate.
\end{remark}

\begin{remark}[Dependence on the steering parameter $\theta_k$]\label{remark2}
    The steering angle $\theta_k,k\in\cK$ affects the weighted sum-rate gap through $|\sin\theta_k|$ and $|\sin\theta_k|^2$ in \eqref{eq:theorem_1_gap} and \eqref{eq:T_1}. This dependence shows a difference between discretization for curved beams and conventional straight beams: the weighted sum-rate gap increases as $|\sin\theta_k|$ increases. This is because when $|\sin\theta_k|$ increases (i.e., the beam becomes more parallel to the array plane), the spatial phase variation along the aperture increases, which requires denser aperture discretization to accurately capture the continuous propagation variation. Consequently, a larger $N$ is required to maintain the same weighted sum-rate gap level when $|\sin\theta_k|$ increases.
\end{remark}

\begin{remark}[Dependence on the focusing parameter $f_k$]\label{remark3}
    The focusing parameter $f_k$ also affects the gap through $1/f_k$ and $1/f_k^2$ in \eqref{eq:T_1} and \eqref{eq:T_2}. This result provides a discretization guideline that is specific to beam focusing. In particular, the weighted sum-rate gap increases as $1/f_k$ increases (i.e., $f_k$ decrease). This is because when $1/f_k$ increases, the beam energy is squeezed into a smaller region, which requires a higher spatial frequency beam across the aperture. Consequently, a larger $N$ is needed to approximate the continuous aperture with the same performance gap when the distance $f_k$ from the transmit array center to the concentrated point decreases.
\end{remark}

\begin{remark}[Dependence on the curvature parameter $B_k$]\label{remark4}
    The curvature parameters $B_k$ contributes to the performance gap via $(2\pi B_k)^3$ and $(2\pi B_k)^6$ in \eqref{eq:T_1} and \eqref{eq:T_2}. Different from straight beam, the curved beam model introduces a cubic phase term over the aperture. Hence, the weighted sum-rate gap increases with $B_k$ grows polynomially. This is because a larger $B_k$ corresponds to a sharper trajectory of the beam, which indicates that the aperture needs to encode stronger nonlinear phase modulation across $\ell$. Consequently, a larger $B_k$ typically requires a larger $N$ to keep the weighted sum-rate loss at the same level.
\end{remark}

\section{Proposed Enhanced BCA Approach}\label{sec:solution}
By using the result in \textbf{Theorem~\ref{theorem:gap}}, we can recast problem \eqref{pro:snr_max} as follows:
\begin{subequations}
\label{pro:dis_snr_max}
\begin{align}
    &\underset{\{\theta_i,f_i,B_i\}_{i=1}^K}{\max}  \ \ \ \sum_{k=1}^K \omega_k \, \tilde R_k\big(\{\theta_i,B_i,f_i\}_{i\in\cK}\big) \tag{\ref{pro:dis_snr_max}}\\
    &\quad \ \ \,    \text {s.t.}  \ \, \qquad \, \,  \sum_{k=1}^K\big\|\bm w_k(\theta_k,f_k,B_k)\big\|_2^2 \Delta_\ell \leq P_\mathrm{T}, \label{cont:dis_power}\\ &\qquad\qquad \qquad\, \eqref{cont:theta},\eqref{cont:f}\ \mathrm{and}\ \eqref{cont:B}.
\end{align}
\end{subequations} 
From \eqref{pro:dis_snr_max}, we observe that the integrations over the continuous aperture in SINR and achievable rate expressions  have been converted into finite summations and vector multiplications via discretization. However, the problem in \eqref{pro:dis_snr_max} is still difficult to solve due to the non-convex objective induced by the fractional structures of the sum-rate function and the SINR expressions. To solve \eqref{pro:dis_snr_max}, we adopt an enhanced BCA approach \cite{Pock2016iPALM,Ahookhosh2021BIBPA}, in which an extrapolation-based surrogate-construction point and proximal regularization are jointly employed to improve update stability and accelerate convergence. This differs from conventional BCA methods, which directly update each block and may therefore suffer from severe oscillations due to the coupling among the optimization variables. In particular, we first reformulate the problem in \eqref{pro:dis_snr_max} via FP, and then decompose it into a sequence of user-wise subproblems which are solved by the proposed enhanced BCA approach. The use of BCA is practical here because the optimization variables can be naturally divided into user-wise blocks \(\{\theta_k,f_k,B_k\}\), and each block update only involves a low-dimensional curved beam parameter optimization problem given the parameters of the other users.


\subsection{Problem Transformation}
To solve \eqref{pro:dis_snr_max}, we use the Lagrangian dual transform in \cite{shenb} to remove the logarithmic function in the achievable rate, i.e., $R_k\big(\{\theta_i,B_i,f_i\}_{i\in\cK}\big)$ for user $k\in\cK$ and recast it as
\begin{align}\label{eq:tilede_R}
\tilde{R}_k=\log(1+\mu_k)&-\mu_k\nonumber\\&+\frac{(1+\mu_k)|u_{k,k}(\theta_k,f_k,B_k)|^2}{\sum_{i=1}^K| u_{i,k}(\theta_i,f_i,B_i)|^2+\sigma_k^2},
\end{align}
where
\begin{align}
    u_{k,i}(\theta_i,f_i,B_i)\triangleq \Delta_\ell\tilde{\bm h}_k^\mathsf{H}\bm w_i(\theta_i,f_i,B_i), \quad k,i\in\cK,
\end{align}
and $\mu_k$ is an auxiliary variable. 

Then, we use the quadratic transform in \cite{shenb} to further transform the fractional terms, i.e., $\frac{(1+\mu_k)|u_{k,k}(\theta_k,f_k,B_k)|^2}{\sum_{i=1}^K| u_{i,k}(\theta_i,f_i,B_i)|^2+\sigma_k^2}$,  in \eqref{eq:tilede_R} into a quadratic functions as
\begin{align}\label{eq:hatR}
    \hat{R}_k=2&\sqrt{1+\mu_k}\Re\big\{\eta_k^*u_{k,k}(\theta_k,f_k,B_k)\big\}
    \nonumber\\
    & 
    -|\eta_k|^2
    \Big(
    \sum_{i=1}^K |u_{k,i}(\theta_i,f_i,B_i)|^2+\sigma_k^2
    \Big),
\end{align}
where $\eta_k$ is an auxiliary variable. 

Hence, the problem in \eqref{pro:dis_snr_max} is transformed and simplified as
\begin{subequations}
\label{pro:fp_transform_problem}
\begin{align}
    \max_{\substack{\{\theta_i,f_i,B_i\}_{i=1}^K,\\\bm\mu, \bm\eta}}
    \ \
    &\,\mathcal{F}\Big(\{\theta_i,f_i,B_i\}_{i=1}^K,\bm{\mu},\bm{\eta}\Big):=\sum_{k=1}^K \frac{\omega_k}{\ln 2}\hat{R}_k
    \tag{\ref{pro:fp_transform_problem}}\\
    \text{s.t.}\quad \ \quad \
    &\,\eqref{cont:theta},\eqref{cont:f}, \eqref{cont:B}\ \mathrm{and}\ \eqref{cont:dis_power},
\end{align}
\end{subequations}
where $\boldsymbol{\mu}\triangleq[\mu_1,\ldots,\mu_K]^\top$ and $\boldsymbol{\eta}\triangleq[\eta_1,\ldots,\eta_K]^\top$.

The objective function in \eqref{pro:fp_transform_problem} remains non-concave with respect to $\{\theta_i,B_i,f_i\}_{i\in\cK}$ due to the nonlinear dependence of the beamforming vectors ${\bm w_i(\theta_i,f_i,B_i)}$ on the physical beam parameters, as well as the coupled multi-user interference terms in the objective function. To tackle this issue, we next decompose \eqref{pro:fp_transform_problem} into a sequence of user-wise subproblems and solve them by an enhanced BCA approach.

\subsection{Proposed Enhanced BCA Approach}
From \eqref{pro:fp_transform_problem}, we see that only the desired-signal term $u_{k,k}(\theta_k,f_k,B_k)$ and the interference terms $\big|\tilde{\bm h}_i^\mathsf{H}\bm w_k(\theta_k,f_k,B_k)\big|^2, \forall i\in\cK$, depend on $\{\theta_k,f_k,B_k\}$. Therefore, we decompose the original problem in \eqref{pro:fp_transform_problem} into $K$ individual subproblem as
\begin{subequations}
\label{pro:block_k_problem}
\begin{align}
    &\underset{\substack{\theta_k,f_k,B_k,\\\mu_k, \eta_k}}{\max}  \qquad \ \ \mathcal{\tilde F}_k(\theta_k,f_k,B_k) \tag{\ref{pro:block_k_problem}}\\
    &\quad \,     \text {s.t.}  \  \qquad \ \ \ -\frac{\pi}{2}\le \theta_k\le \frac{\pi}{2},\label{cont:t_k2}\\ &\qquad\qquad \qquad 0\leq f_k\le x_k,\label{cont:f_k2}\\ &\qquad\qquad \ \ \ \ \ \ -B_k^{\max}\leq B_k\leq B_k^{\max}, \label{cont:B_k2}
\end{align}
\end{subequations}
where
\begin{align}
    &\mathcal{\tilde F}_k(\theta_k,f_k,B_k)
    \triangleq \frac{2\omega_k}{\ln 2}\sqrt{1+\mu_k}
\Re\!\left\{\eta_k^*\,u_{k,k}(\theta_k,f_k,B_k)\right\}
\nonumber\\
&\qquad\qquad\qquad\quad\
-\sum_{i=1}^K\frac{\omega_i}{\ln 2}|\eta_i|^2\Delta_\ell^2
\Big|\tilde{\bm h}_i^\mathsf{H}\bm w_k(\theta_k,f_k,B_k)\Big|^2 .
\end{align}
However, \eqref{pro:block_k_problem} is still non-concave, since the beamforming vector $\bm w_k(\theta_k,f_k,B_k)=\bm\alpha_k\odot e^{j\bm\beta_k(\theta_k,f_k,B_k)}$ for user $k$ is non-concave with respect to $\theta_k$, $f_k$, and $B_k$.

Next, to solve each above non-concave subproblem in \eqref{pro:block_k_problem} for user $k\in\cK$, we construct a tractable surrogate function for the objective function $\tilde{\mathcal F}_k(\theta_k,f_k,B_k)$. Combined with the preceding FP transformation in \eqref{eq:tilede_R}--\eqref{eq:sr_term2_new}, this allows the overall problem to be tackled within an minorization–maximization framework, where the auxiliary variables $\bm \mu, \bm \eta$ and the variables $\{\theta_k,f_k,B_k \}, k\in\cK$ are updated iteratively until convergence to a stationary solution.


To construct the surrogate function, we first need to determine the surrogate-construction point. Different from current works that directly use the solution of last iteration to construct surrogate function, we generate the surrogate-construction point based on the solutions at two previous iterations \cite{Pock2016iPALM}. Here, using the information from two previous iterations, can capture the local descent trend thus accelerating convergence, i.e., \begin{align}
    \tilde{\theta}_k
    &=
    \underline{\dot{\theta}_k}
    +
    \upsilon_k\big(\,\underline{\dot{\theta}_k}-\underline{\ddot{\theta}_k}\,\big),\label{eq:theta_bar}\\
    \tilde{f}_k
    &=
    \underline{\dot{f}_k}
    +
    \upsilon_k\big(\,\underline{\dot{f}_k}-\underline{\ddot{f}_k}\,\big),\label{eq:f_bar}\\
    \tilde{B}_k
    &=
    \underline{\dot{B}}_k
    +
    \upsilon_k\big(\,\underline{\dot{B}_k}-\underline{\ddot{B}_k}\,\big),\label{eq:B_bar}
\end{align}
where $\{\underline{\dot{\theta}_k,\dot{f}_k,\dot{B}_k}\}$ and $\{\underline{\ddot{\theta}_k,\ddot{f}_k,\ddot{B}_k}\}$ 
are the two previous points, $\upsilon_k\in[0,1)$ is a parameter which determines how weight of the two previous points contributing to the update. 


Based on the surrogate-construction point in \eqref{eq:theta_bar}--\eqref{eq:B_bar}, we construct a linear surrogate function of the objective function in problem \eqref{pro:block_k_problem}. For notational compactness, we rewrite the optimization variables and the surrogate-construction point in vector forms as
\begin{align}
    \bm{c}_k \triangleq [\theta_k,f_k,B_k]^\top,
    \qquad
    \tilde{\bm{c}}_k \triangleq \big[\tilde{\theta}_k,\tilde{f}_k,\tilde{B}_k\big]^\top.
\end{align}
Given $\bm{c}_k$ and $\tilde{\bm{c}}_k$, by employing the first-order \textit{Taylor expansion}, the linear surrogate function of the objective function in problem \eqref{pro:block_k_problem} is
\begin{align}\label{eq:surrogate_1}
     \hat{\mathcal F}_k(\bm{c}_k\mid \tilde{\bm{c}}_k)\triangleq\tilde{\mathcal F}_k(\tilde{\bm{c}}_k)
    +\nabla \tilde{\mathcal F}_k(\tilde{\bm{c}}_k)^\top
    (&\bm{c}_k-\tilde{\bm{c}}_k),
\end{align}
where 
\begin{align}
    \nabla \tilde{\mathcal F}_k(\tilde{\bm{c}}_k)^\top
    =
    \left[
    \frac{\partial \mathcal{\tilde F}_k({\bm{c}}_k)}{\partial \theta_k}\bigg|_{\tilde{\bm{c}}_k},
    \frac{\partial \mathcal{\tilde F}_k({\bm{c}}_k)}{\partial f_k}\bigg|_{\tilde{\bm{c}}_k},
    \frac{\partial \mathcal{\tilde F}_k({\bm{c}}_k)}{\partial B_k}\bigg|_{\tilde{\bm{c}}_k}
    \right]^\top,
\end{align}
denotes the gradient of $\mathcal{\tilde F}_k(\theta_k,f_k,B_k)$ at the surrogate-construction point $\tilde{\bm{c}}_k$.

However, the above linear surrogate function in \eqref{eq:surrogate_1} may lead to large variable variations between consecutive iterations and thus causing numerical oscillations, since the surrogate-construction point $\tilde{\bm c}_k$ is obtained by extrapolating the two previous points. Therefore, to restrict large update steps and improve stability, we further incorporate a quadratic proximal term into the linear surrogate function in \eqref{eq:surrogate_1} \cite{Pock2016iPALM,Ahookhosh2021BIBPA}, i.e.,
\begin{align}\label{eq:surrogate_2}
    \bar{\mathcal F}_k(\bm c_k\mid \tilde{\bm c}_k)
    \triangleq
    \hat{\mathcal F}_k(\bm{c}_k\mid \tilde{\bm{c}}_k)
    -\frac{\rho_k}{2}\|\bm c_k-\tilde{\bm c}_k\|_2^2,
\end{align}
where $\kappa_k>0$ is the proximal parameter. Next, we prove that the function $\bar{\mathcal F}_k(\bm c_k\mid \tilde{\bm c}_k)$ in \eqref{eq:surrogate_2} still remains the surrogate function of the original objective function in \eqref{pro:block_k_problem}.

\begin{proposition}
$\bar{\mathcal F}_k(\bm c_k\mid \tilde{\bm c}_k)$ remains a valid surrogate function $\tilde{\mathcal F}_k(\bm c_k)$ if the proximal parameter $\rho_k$ is chosen such that
\begin{align}\label{eq:con_rho}
    \rho_k \ge 
    \sup_{\substack{\bm x,\bm y\in\mathcal X_k\\ \bm x\neq \bm y}}
    \frac{\|\nabla \tilde{\mathcal F}_k(\bm x)-\nabla \tilde{\mathcal F}_k(\bm y)\|_2}
    {\|\bm x-\bm y\|_2},
\end{align}
where $\mathcal X_k$ represents the feasible set of $\bm c_k$.
\end{proposition}

\begin{proof}
To prove that $\bar{\mathcal F}_k(\bm c_k\mid \tilde{\bm c}_k)$ is a valid surrogate function of $\tilde{\mathcal F}_k(\bm c_k)$ based on \cite{Nesterov2004,Beck2017FOM}, we must verify the following three conditions: i) $\bar{\mathcal F}_k(\tilde{\bm{c}}_k\mid \tilde{\bm{c}}_k)=\tilde{\mathcal F}_k(\tilde{\bm{c}}_k)$, ii) $\nabla_{\bm{c}_k}\bar{\mathcal F}_k(\tilde{\bm{c}}_k\mid \tilde{\bm{c}}_k)
    =
    \nabla_{\bm{c}_k}\tilde{\mathcal F}_k(\tilde{\bm{c}}_k)$, and iii) $\bar{\mathcal F}_k({\bm{c}}_k\mid \tilde{\bm{c}}_k)\leq \tilde{\mathcal F}_k({\bm{c}}_k)$.

\textbf{i) }To verify the first condition, we substitute $\bm c_k=\tilde{\bm c}_k$ into \eqref{eq:surrogate_2}. Hence, we have $\bar{\mathcal F}_k(\tilde{\bm c}_k\mid \tilde{\bm c}_k)
    =
    \tilde{\mathcal F}_k(\tilde{\bm c}_k)$ as well as condition i) is verified.

\textbf{ii) }To verify the condition ii), we differentiate \eqref{eq:surrogate_2} with respect to $\bm c_k$ yields 
\begin{align}\label{eq:grad_F}
    \nabla_{\bm c_k}\bar{\mathcal F}_k(\bm c_k\mid \tilde{\bm c}_k)
    =
    \nabla \tilde{\mathcal F}_k(\tilde{\bm c}_k)
    -\rho_k(\bm c_k-\tilde{\bm c}_k).
\end{align}
By further setting $\bm c_k=\tilde{\bm c}_k$ in \eqref{eq:grad_F}, we have $\nabla_{\bm c_k}\bar{\mathcal F}_k(\tilde{\bm c}_k\mid \tilde{\bm c}_k)
    =
    \nabla_{\bm c_k}\tilde{\mathcal F}_k(\tilde{\bm c}_k)$. Hence, condition ii) is satisfied. 

\textbf{iii) }Finally, to satisfy condition iii),  since $\tilde{\mathcal F}_k(\bm c_k)$ is twice continuously differentiable on the compact feasible set $\mathcal X_k=[-\pi/2,\pi/2]\times[0,x_k]\times[-B_k^{\max},B_k^{\max}]$, the gradient $\nabla \tilde{\mathcal F}_k(\bm c_k)$ is \textit{Lipschitz continuous} on $\mathcal X_k$ \cite{Nesterov2004,Beck2017FOM}. Therefore, there exists a constant $\rho_k>0$ such that 
\begin{align}
    \rho_k \ge 
    \sup_{\substack{\bm x,\bm y\in\mathcal X_k\\ \bm x\neq \bm y}}
    \frac{\|\nabla \tilde{\mathcal F}_k(\bm x)-\nabla \tilde{\mathcal F}_k(\bm y)\|_2}
    {\|\bm x-\bm y\|_2}.
\end{align}
Then, according to the \textit{standard descent lemma} \cite[Lemma 4.22]{Beck2017FOM}, for any $\bm c_k\in\mathcal X_k$, we have 
\begin{align}
    \tilde{\mathcal F}_k(\bm c_k)
    &\ge\,
    \tilde{\mathcal F}_k(\tilde{\bm c}_k)
    +\nabla \tilde{\mathcal F}_k(\tilde{\bm c}_k)^\top
    (\bm c_k-\tilde{\bm c}_k)
    -\frac{\rho_k}{2}\|\bm c_k-\tilde{\bm c}_k\|_2^2\nonumber\\
    &=\bar{\mathcal F}_k(\bm c_k\mid \tilde{\bm c}_k),
    \qquad \forall \bm c_k\in\mathcal X_k.
\end{align}
This verifies condition iii) and comples the proof. 
\end{proof}

Given \eqref{eq:con_rho}, the non-concave problem \eqref{pro:block_k_problem} is recast as
\begin{subequations}
\label{eq:surrogate_thetafB}
\begin{align}
    &\underset{\bm c_k}{\max}  \qquad \ \  \bar{\mathcal F}_k({\bm{c}}_k\mid \tilde{\bm{c}}_k) \tag{\ref{eq:surrogate_thetafB}}\\
    &\     \text {s.t.}  \  \quad \ \ \  \ \ \bm c_k \in \mathcal{X}_k. 
\end{align}
\end{subequations}
Since problem \eqref{eq:surrogate_thetafB} is a concave quadratic problem, the optimal update of the curving parameters $\bm c_k=[\theta_k,f_k,B_k]^\top, k\in\cK$, is obtained in closed-form from the first-order condition, i.e., $\partial\bar{\mathcal F}_k(\bm{c}_k| \tilde{\bm{c}}_k)/\partial\bm c_k=0$. Then, by projecting the unconstrained update onto the feasible set defined by \eqref{cont:t_k2}--\eqref{cont:B_k2}, we have
\begin{align}
    \theta_k^{\star}
    &=\min\bigg\{ \max\bigg\{ \frac{1}{\rho_{k}}\times
    \frac{\partial \tilde{\mathcal F}_k}{\partial\theta_k}\bigg|_{\tilde{\bm c}_k}, -\frac{\pi}{2}\bigg\},\frac{\pi}{2} \bigg\},\label{eq:theta_update}\\
    f_k^{\star}
    &=
    \min\bigg\{ \max\bigg\{ \frac{1}{\rho_{k}}\times
    \frac{\partial \tilde{\mathcal F}_k}{\partial f_k}\bigg|_{\tilde{\bm c}_k}, 0\bigg\},x_k \bigg\},\label{eq:f_update}\\
    B_k^{\star}
    &=
    \min\bigg\{ \max\bigg\{ \frac{1}{\rho_{k}}\times
    \frac{\partial \tilde{\mathcal F}_k}{\partial B_k}\bigg|_{\tilde{\bm c}_k}, -B_k^{\max}\bigg\},B_k^{\max} \bigg\}. \label{eq:B_update}
\end{align}

Next, we update the auxiliary variables $\bm\mu$ and $\bm\eta$. Given $\{\theta_k,f_k,B_k\}_{i\in\cK}$, the objective function in \eqref{eq:tilede_R} and \eqref{eq:hatR} are convex with respect to $\bm\mu$ and $\bm\eta$. Thus, the optimal values of each $\mu_k$ and $\eta_k$ can be calculated by the first-order condition, i.e., $\partial \tilde{R}/\partial \mu_k=0$ and $\partial \hat{R}/\partial \eta_k=0, \forall k\in\cK$. Hence, we have
\begin{align}
    &\mu_k^\star=\gamma_k\big(\{\theta_i,B_i,f_i\}_{i\in\cK}\big), \forall k\in\cK,\label{eq:mu}\\
    &\eta_k^\star
    =
    \frac{
    \sqrt{1+\mu_k}\,u_{k,k}(\theta_k,f_k,B_k)
    }{
    \sum_{i=1}^K |u_{k,i}(\theta_i,f_i,B_i)|^2+\sigma_k^2
    }, \forall k\in\cK.
    \label{eq:sr_term2_new}
\end{align}

The auxiliary variables $\bm \mu$ and $\bm \eta$, together with the optimization variables $\{\theta_k,f_k,B_k \}, k\in\cK$ are updated iteratively until convergence to a stationary solution. 

After obtaining the curving parameters
$\{\theta_i^\star,f_i^\star,B_i^\star\}_{i\in\cK}$, we fix the corresponding
phase profiles and design the amplitude vectors
$\{\bm\alpha_i\}_{i\in\cK}$ to meet the power constraint in \eqref{cont:dis_power} according to the effective channel gains.
In particular, we define the normalized beam for user $k\in\cK$ as
\begin{align}
    \bar{\bm w}_k
    \triangleq (1/\sqrt{L})\times e^{j\bm\beta_k(\theta_i^\star,f_i^\star,B_i^\star)},
    \qquad \forall k\in\cK,
\end{align}
which satisfies $\Delta_\ell\|\bar{\bm w}_k\|_2^2=1$.
Then, we define the effective channel gain for user $k\in\cK$ as $G_k
    \triangleq\Delta_\ell^2
    \big|\bar{\bm w}_k^\mathsf{H}\tilde{\bm h}_k\big|^2$. 
Based on these gains, the transmit power is allocated proportionally as
\begin{align}
    p_k=\frac{P_{\mathrm T}G_k}{\sum_{i=1}^K G_i},
    \qquad \forall k\in\cK,
\end{align}
such that $\sum_{k=1}^K p_k=P_{\mathrm T}$.
Accordingly, we set
\begin{align}\label{eq:update_alpha}
    \bm\alpha_k^\star
    =
    \sqrt{\frac{p_k}{L}}\bm 1_N
    =
    \sqrt{
    \frac{P_{\mathrm T}G_k}{L\sum_{i=1}^K G_i}
    }\bm 1_N,
    \qquad \forall k\in\cK,
\end{align}
where $\bm 1_N\in\mathbb{R}^{N\times1}$ is the all-one vector.
Therefore, the final beamforming vector is
\begin{align}
    \bm w_k^\star
    =
    \bm\alpha_k^\star
    \odot
    e^{j\bm\beta_k(\theta_k^\star,f_k^\star,B_k^\star)},
    \qquad \forall k\in\cK.
\end{align}

The overall procedure of the proposed enhanced BCA approach for solving problem \eqref{pro:dis_snr_max} is summarized in Algorithm~\ref{alg:PCA}.

\subsection{Implementation and Complexity Analysis}
 Finally, we analyze the implementation and complexity of the proposed method, which are specified as follows:
\subsubsection{Implementation analysis}
The implementation of the proposed method consists of four stages: i) update the auxiliary variables $\bm\mu$ and $\bm\eta$, ii) find the surrogate-construction point $\{\tilde{\theta}_k,\tilde{f}_k,\tilde{B}_k\}_{k\in\cK}$, iii)  optimize $\{{\theta}_k,{f}_k,{B}_k\}_{k\in\cK}$ with fixed $\bm\mu$, $\bm\eta$ and $\{\tilde{\theta}_k,\tilde{f}_k,\tilde{B}_k\}_{k\in\cK}$ using the enhanced BCA method, and iv) determine the amplitude $\{\bm\alpha_k\}_{k\in\cK}$.

To update the auxiliary variables, we need the locations of all users $\bm r_k, \forall k$ and the blockage region $\Omega_o$, the power budget $P_{\mathrm{T}}$, and the Gaussian noise power $\sigma_k^2,\forall k$. Meanwhile, we randomly initialize the optimization variables, i.e., $\{\tilde{\theta}_k,\tilde{f}_k,\tilde{B}_k\}_{k\in\cK}$, to the feasible set. Given this information, we use \eqref{eq:mu} and \eqref{eq:sr_term2_new} to update the auxiliary variables $\bm\mu$ and $\bm\eta$ induced by the Lagrangian dual transform and quadratic transform. Then, we use \eqref{eq:theta_bar} to \eqref{eq:B_bar} to generate a surrogate-construction point and construct a surrogate function around this point by \eqref{eq:surrogate_2}. Next, we optimize $\{{\theta}_k,{f}_k,{B}_k\}_{k\in\cK}$ using the closed-forms by \eqref{eq:theta_update}--\eqref{eq:B_update}. By repeating the above procedures, we obtain the optimal $\{\tilde{\theta}_k,\tilde{f}_k,\tilde{B}_k\}_{k\in\cK}$ until the objective converges. Finally, we use \eqref{eq:update_alpha} to determine the amplitude $\{\bm\alpha_k\}_{k\in\cK}$ as well as the power allocation.

The proposed method is summarized in Algorithm~\ref{alg:PCA}, and the per-iteration variable update is given by \eqref{eq:update_flow}.
\begin{equation}\label{eq:update_flow}
\underbrace{%
  \bm\mu \to\bm\eta
   \to
  \{\tilde{\theta}_k,\tilde{f}_k,\tilde{B}_k\}_{k\in\cK}\to\{\tilde{\theta}_k,\tilde{f}_k,\tilde{B}_k\}_{k\in\cK}
}_{\text{one round of variable update}}\to 
  \{\bm\alpha_k\}_{k\in\cK}.
\end{equation}

\begin{figure*}[t]
    \centering
    \subfloat[\scriptsize Straight beam in single-user case.]{
        \includegraphics[width=0.23\textwidth]{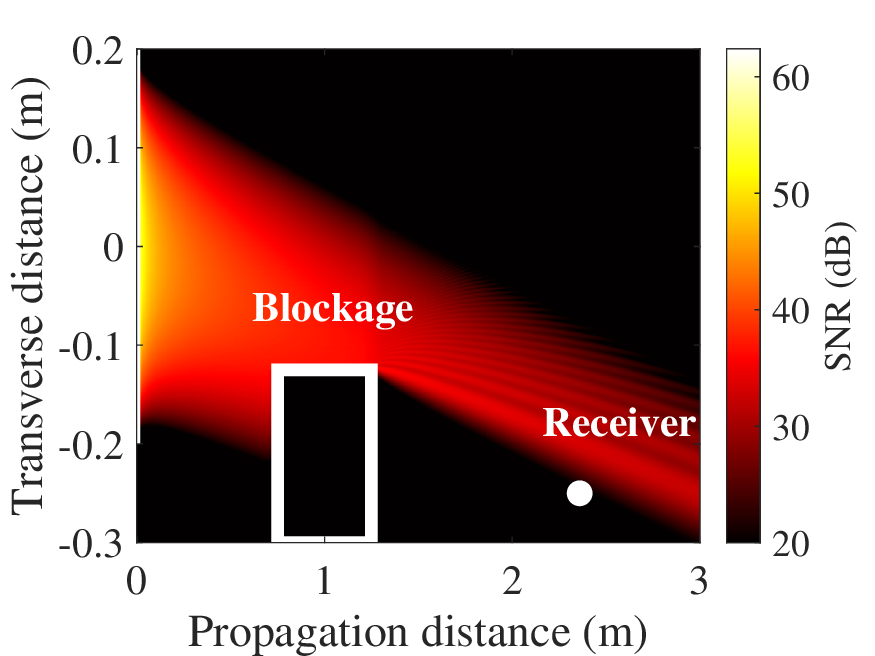}
        \label{fig:demo_1}
    }
        \hfill
    \subfloat[\scriptsize Curved beam in single-user case.]{
        \includegraphics[width=0.23\textwidth]{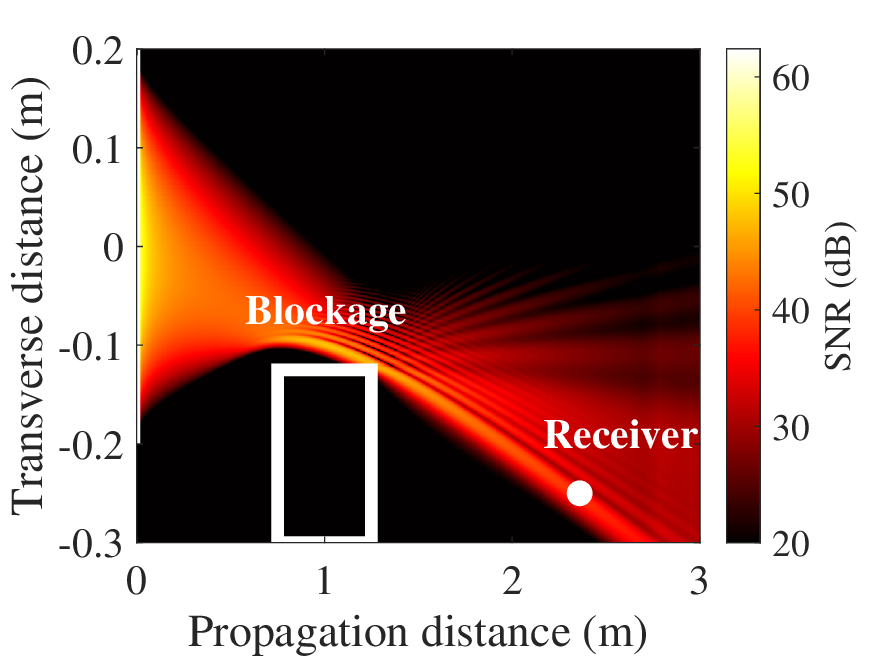}
        \label{fig:demo_2}
    }
        \hfill
    \subfloat[\scriptsize Straight beams in multi-user case.]{
        \includegraphics[width=0.23\textwidth]{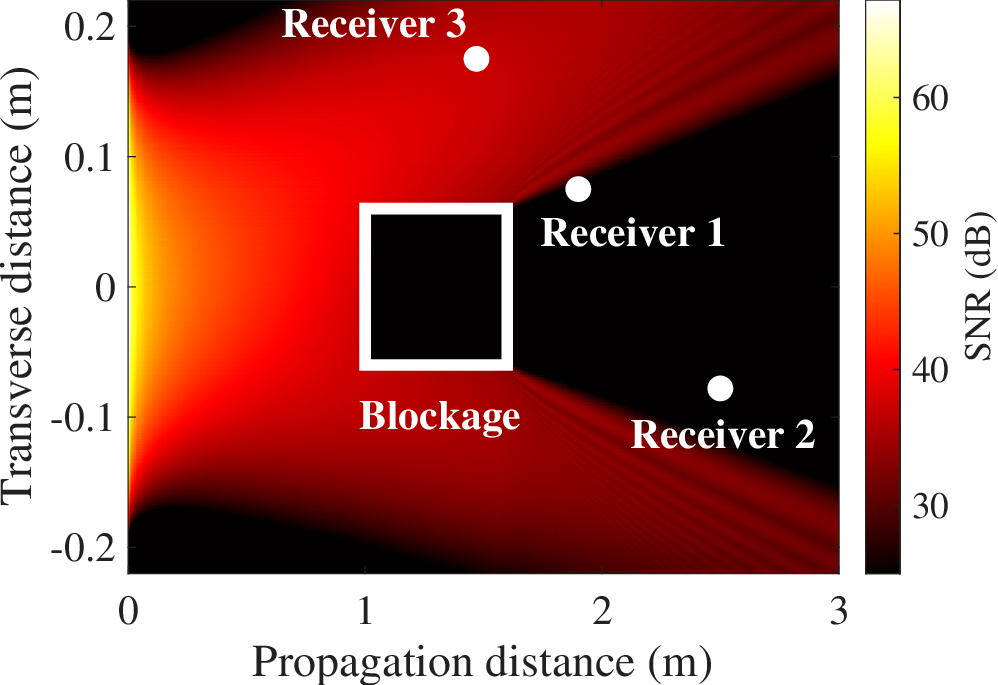}
        \label{fig:demo_3}
    }
     \hfill
    \subfloat[\scriptsize Curved beams in multi-user case.]{
        \includegraphics[width=0.23\textwidth]{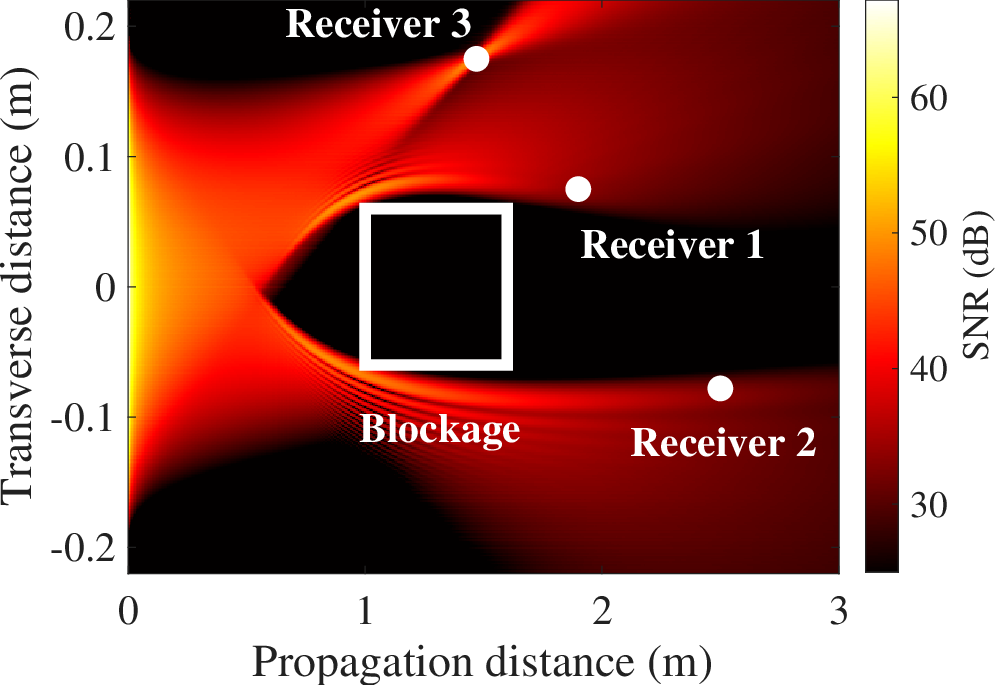}
        \label{fig:demo_4}
    }
    \caption{Visualizations of straight and curved beams in single-user and multi-user cases.}
    \label{fig:demo}\vspace{-1.4em}
\end{figure*}

\subsubsection{Complexity analysis}
The complexity analysis of the proposed approach consists of three main components:i) updating the auxiliary variables, ii) updating the optimization variables, iii) determining the amplitudes. 

\begin{algorithm}[t]
\caption{Enhanced BCA for solving problem \eqref{pro:dis_snr_max}}
\label{alg:PCA}
\DontPrintSemicolon
\SetAlgoLined
\KwIn{$\tilde{\bm h}_k,\forall k\in\cK$}
\KwOut{$\{\theta_k^\star,f_k^\star,B_k^\star\}$ and $\bm\alpha_k, k\in\cK$}

Initialize $\{\theta_k,f_k,B_k,\bm \alpha_k\}, k\in\cK$ to feasible values. \;

\Repeat{$\mathrm{the\ objective\  value\ converges}$}{

Update each $\mu_k$ by \eqref{eq:mu}.

Update each $\eta_k$ by \eqref{eq:sr_term2_new}.

Update each $\{\tilde{\theta}_k,\tilde{f}_k,\tilde{B}_k\}$ by \eqref{eq:theta_bar}--\eqref{eq:B_bar}.

Update each $\theta_k$ by \eqref{eq:theta_update}.

Update each $f_k$ by \eqref{eq:f_update}.

Update each $B_k$ by \eqref{eq:B_update}.

}

Update each $\bm \alpha_k$ by \eqref{eq:update_alpha}.

\end{algorithm}

i) \textit{\textbf{Complexity of updating the auxiliary variables:}} The overall update consists of calculating $\mu_k$ and $\eta_{k}$. The computation of $\mu_k$ and $\eta_k$ per iteration is primarily determined by the sum of the norm squares $\sum_{i=1}^K|\bu_{k,i}(\theta_i,f_i,B_i)|^2$, which has a complexity of $\mathcal{O}(NK)$. Since we have $K$ variables of $\mu_k$, the overall complexity of calculating $\mu_k$ is $\mathcal{O}(K^2N)$. Updating $\eta_k$ is similar to the update of $\mu_k$, which induces a complexity of $\mathcal{O}(K^2N)$ due to the the norm squares $\sum_{i=1}^K|\bu_{k,i}(\theta_i,f_i,B_i)|^2$. Thus, the overall complexity of updating all auxiliary variables is $\mathcal{O}(K^2N)$.

ii) \textit{\textbf{Complexity of updating the optimization variables:}} 
For each user \(k\), updating the optimization variables $\{\theta_k,f_k,B_k\}$ requires evaluating the gradient of the objective function \(\tilde{\mathcal{F}}_k(\theta_k,f_k,B_k)\) with respect to the three variables at the surrogate-construction point. The complexity is mainly determined by calculating the desired and interference terms \(\{u_{i,k}(\theta_k,f_k,B_k)\}_{i=1}^K\), where each term involves an inner product over \(N\) discrete aperture samples and thus has complexity \(\mathcal{O}(N)\). Since there are \(K\) such terms for each user \(k\), the complexity of computing the gradient for one user is \(\mathcal{O}(KN)\). After obtaining the gradient, the update of \(\theta_k\), \(f_k\), and \(B_k\) is performed by closed-form projection, which incurs \(\mathcal{O}(1)\) complexity. Therefore, the overall complexity of updating the optimization variables for all users is \(\mathcal{O}(K^2N)\).

iii) \textit{\textbf{Complexity of determining the amplitudes:}} 
After updating the optimization variables, the amplitudes are determined according to the effective channel gains \(G_k=\Delta_\ell^2\big|\bar{\mathbf{w}}_k^H\tilde{\mathbf{h}}_k\big|^2\), \(k\in\mathcal{K}\). Computing each \(G_k\) requires one inner product of two \(N\)-dimensional vectors, which has complexity \(\mathcal{O}(N)\). For a total of \(K\) users, the total complexity of calculating $G_k,\forall k\in\cK$ is \(\mathcal{O}(KN)\). Then, the power allocation \(p_k=\frac{P_TG_k}{\sum_{i=1}^K G_i}\) requires an additional \(\mathcal{O}(K)\) complexity. Finally, determining the amplitude vectors \(\boldsymbol{\alpha}_k=\sqrt{p_k/L}\times\mathbf{1}_N\), \(k\in\mathcal{K}\), requires \(\mathcal{O}(N)\) for each user, which in turn resulting \(\mathcal{O}(KN)\) complexity for all users. Hence, the overall complexity of determining the amplitudes is \(\mathcal{O}(KN)\).

Let $T$ be the maximum number of iterations, the overall complexity of the proposed approach is $\mathcal{O}\big(TK^2N\big)$, which is polynomial and scales linearly with the number \(N\) of aperture samples for a fixed number $K$ of receivers. For comparison, a commonly used zero-forcing beamforming method in practical multiuser networks has a computational complexity of \(\mathcal{O}(K^2N+K^3)\), which is also polynomial and scales linearly with \(N\). Hence, the proposed approach is computationally practical for the considered multiuser networks.

\section{Simulation Results and Analysis}\label{sec:experiments}
Next, we numerically evaluate and analyze the considered curved beam model and proposed approach. We consider a near-field scenario (e.g., indoor environment) where the blockage region (e.g., partition wall or large furniture) is located within \([0.5,1]\times[-0.4,0.4]\) m\(^2\) and receivers are randomly distributed within \([1,3]\times[-0.5,0.5]\) m\(^2\). The simulation settings are specified in Table~\ref{tab:parameters} unless otherwise stated. Moreover, the carrier frequency is $300$ GHz, i.e., a wavelength of $\lambda=10^{-3}$ m, which lies at the boundary between the mmWave and THz bands. The large-scale pathloss is $\epsilon_k=10^{-\mathrm{PL}_k/10}$, where $\mathrm{PL}_k=72+29.2\log \|\bm p_k\|, \forall k\in\cK$ \cite{pathloss}. For comparisons, we consider the following benchmarks:  \textbf{\textit{i) Straight beam}}: This benchmark only optimizes the steering angle $\{\theta_k\}_{k\in\cK}$ by greedy search. This scheme is used to show the performance gain achieved by the standard straight beams \cite{Cur-THz-conf,Cur-THz-NC,VanTrees2002}.  \textit{\textbf{ii) Greedy search}}: This benchmark updates $\{\theta_k,f_k,B_k\}_{k\in\cK}$ via greedy search. This scheme is used to provide an approximate upper bound on the performance gain achievable by curved beams. \textit{\textbf{iii) Gradient ascent}}: 
This benchmark directly updates the optimization variables $\{\theta_k,f_k,B_k \}_{k\in\cK}$ by gradient descent based on the original objective function, without employing the surrogate function. \textit{\textbf{iv) Conventional BCA}}: This benchmark optimizes the optimization variables $\{\theta_k,f_k,B_k \}_{k\in\cK}$ by directly constructing surrogate function based on the solution of last iteration without the proximal regularization term. This scheme is used to show the improvement achieved by the proposed method over conventional BCA. \textit{\textbf{v) CAPA}}: This benchmark adopts the PDM-based design in \cite{C-MIMO}, which transforms the continuous pattern design into the optimization of projection coefficients on finite orthogonal bases. This scheme is used to compare the proposed curved beams scheme with a conventional CAPA optimization method.



Fig.~\ref{fig:demo} shows four examples of straight and curved beams in the single-user and multi-user cases. In Figs.~\ref{fig:demo_1} to~\ref{fig:demo_4}, a blockage is placed in the middle of the propagation region, while the receivers are located behind the blockage. From Fig.~\ref{fig:demo_1} and Fig.~\ref{fig:demo_3}, we observe that the propagation paths are obstructed, and the receivers do not receive useful signals. In contrast, in Fig.~\ref{fig:demo_2} and Fig.~\ref{fig:demo_4}, the SNRs at the receivers are higher than that of the receivers in Fig.~\ref{fig:demo_1} and Fig.~\ref{fig:demo_3}. These results visually demonstrate the effectiveness of curved beams in bypassing blockage and restoring reliable signal reception.



\begin{table}[t]
\renewcommand{\arraystretch}{1}
    \centering
    \caption{ List of Simulation Parameters}
    \vspace{-0.6em}
    \begin{tabular}{|c||c||c||c|}
         \hline
         Parameters& Values & Parameters & Values \\
         \hline
          \hline
         $L$ & $0.2$ m & $N$& $512$ \\
         \hline
         $K$& $4$ & $M$ & $3$ \\
         \hline
         $B^{\max}_k$& $12$ &  $|\xi_m|$ & $10^{-2}$\\
         \hline
         $P_\mathrm{T}$ & $30$ dBm & $\sigma_k^2$& $-90$ dBm\\
         \hline
    \end{tabular}
       \vspace{-2em}
    \label{tab:parameters}
\end{table}

\begin{figure*}[t]
    \centering
    \subfloat[\scriptsize $N=128, L=0.4$]{
        \includegraphics[width=0.31\textwidth]{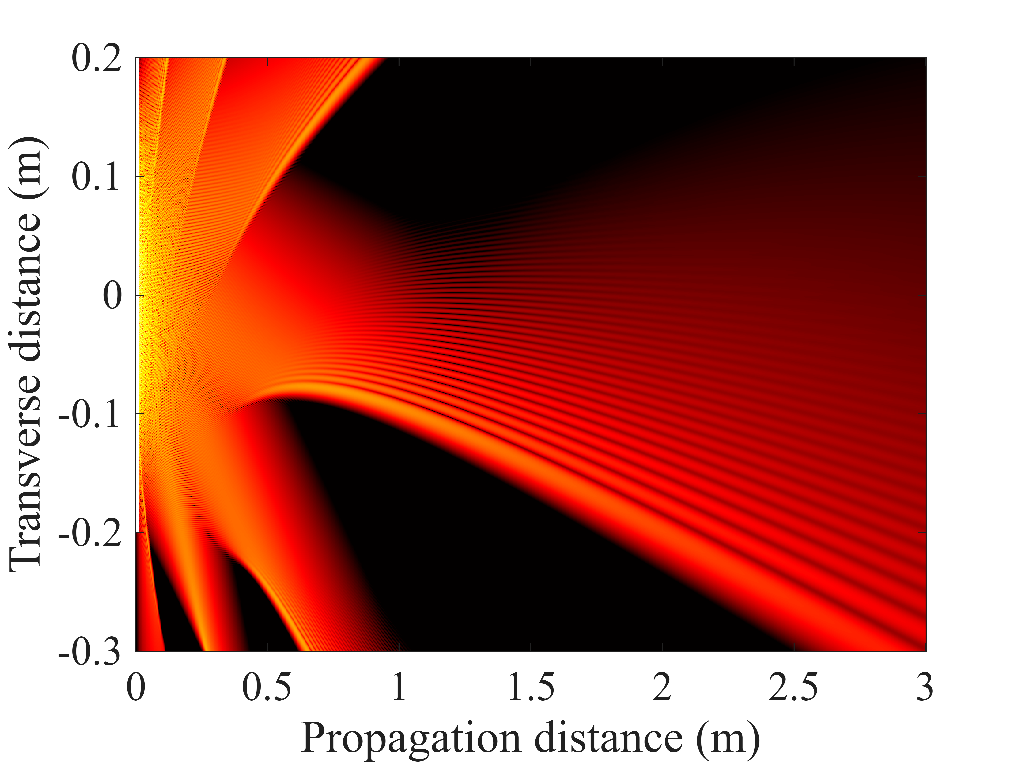}
        \label{fig:n128}
    }\hfill
    \subfloat[\scriptsize $N=512, L=0.4$]{
        \includegraphics[width=0.31\textwidth]{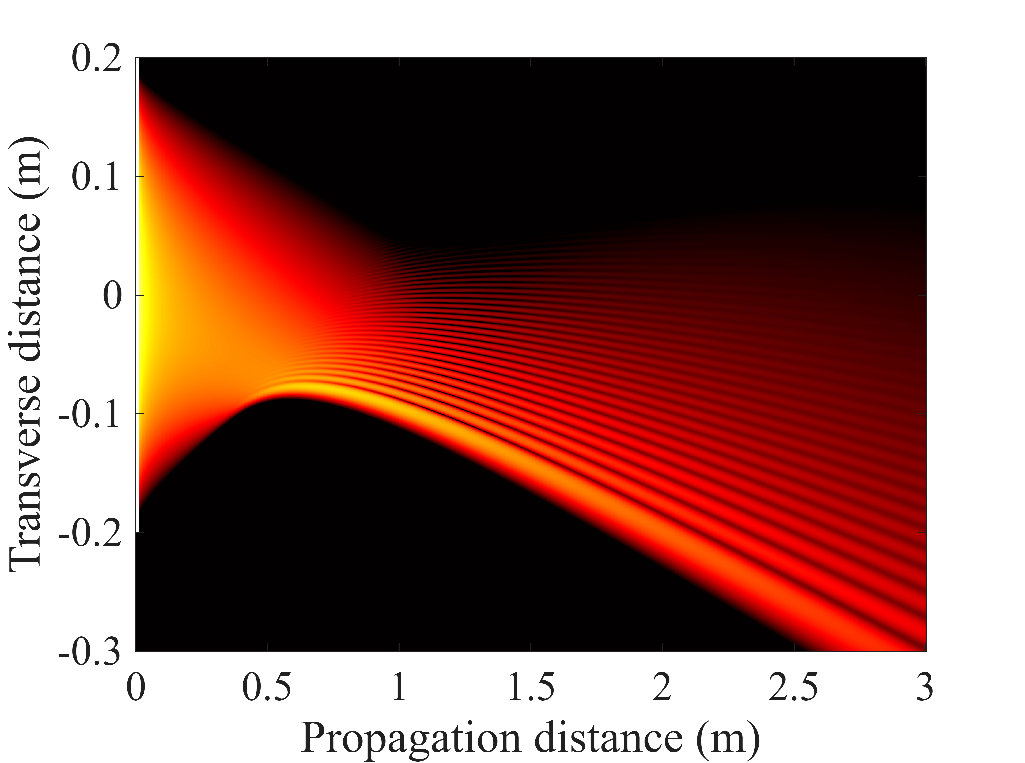}
        \label{fig:n512}
    }\hfill
    \subfloat[\scriptsize $N=512, L=0.6$]{
        \includegraphics[width=0.31\textwidth]{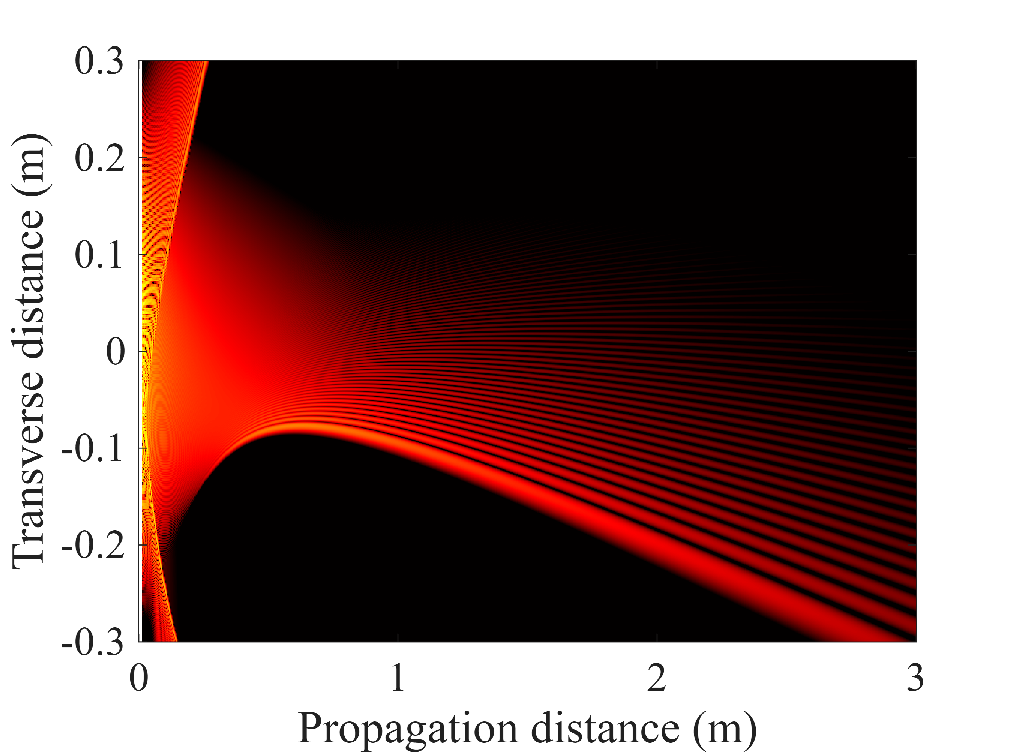}
        \label{fig:L03}
    }
    \caption{Visualizations of curved beams with different number $N$ of discrete samples and different aperture size $L$.}
    \label{fig:different_n} \vspace{-2.2em}
\end{figure*}


\begin{figure}[t]
\centering
\includegraphics[width=6.8cm]{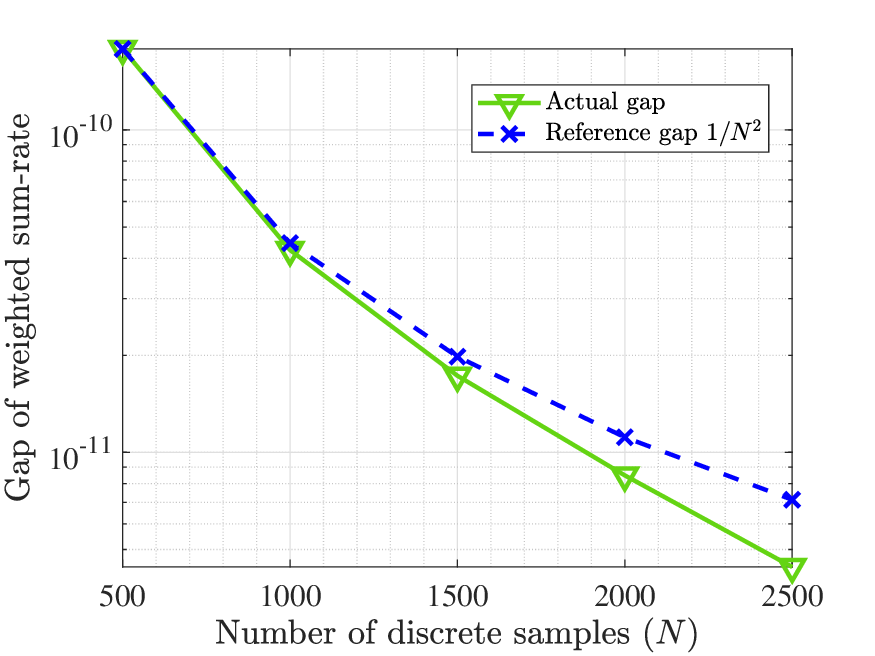}
\vspace{-1em}
\caption{Gap of weighted sum-rates vs. number $N$ of discrete samples.}
\vspace{-1em}
\label{fig:gap}
\end{figure}

Fig.~\ref{fig:different_n} shows how curved beams change as the number $N$ of discrete samples and aperture sizes vary. In Fig.~\ref{fig:n128} and Fig.~\ref{fig:n512}, the aperture size is fixed at $L=0.4$ while the number of samples increases from $N=128$ to $N=512$. It can be observed that a dense discretization can suppress the undesired sidelobe leakage and yields a smooth and accurate curved beam path. This is because a larger $N$ allows the discrete aperture to better approximate the desired continuous aperture, thereby reducing the discretization error. In Fig.~\ref{fig:L03}, $N$ is fixed at $512$ while the aperture size is increased from $L=0.4$ to $L=0.6$. In this case, the energy of the main lobe becomes more concentrated due to the larger aperture size. However, since the number of discrete samples does not increase with the enlarged aperture size, the sampling is still insufficient, which results in a sidelobe leakage.

Fig.~\ref{fig:gap} shows the weighted sum-rate gap versus the number \(N\) of discrete samples under a fixed set of beam parameters \(\{\theta_k,f_k,B_k\}_{k\in\mathcal{K}}\). The actual gap is calculated as the difference between the numerical weighted sum-rate of the continuous aperture array and that of the discrete array. In particular, the weighted sum-rate of the continuous aperture array is evaluated using the Riemann integration tool in MATLAB, while that of the discrete array is computed via a finite summation over \(N\) samples. The reference gap is generated according to the \(1/N^2\) scaling provided by \textbf{Theorem~\ref{theorem:gap}}. From this figure, we see that when \(N=500\), the actual weighted sum-rate gap is around \(10^{-10}\). Fig.~\ref{fig:gap} also shows that the actual gap further decreases as \(N\) increases. This is because a larger $N$ leads to a more accurate approximation of the continuous aperture integration, which reduces undesired sidelobe leakage and improves beamforming performance. Fig.~\ref{fig:gap} also shows that the difference between the actual gap and the reference gap increases as \(N\) increases. This is because, when the gap is sufficiently small (e.g., at the \(10^{-11}\) level), the actual gap becomes more sensitive to higher-order discretization terms beyond the dominant \(1/N^2\) scaling.

\begin{figure}[t]
\centering
\includegraphics[width=6.8cm]{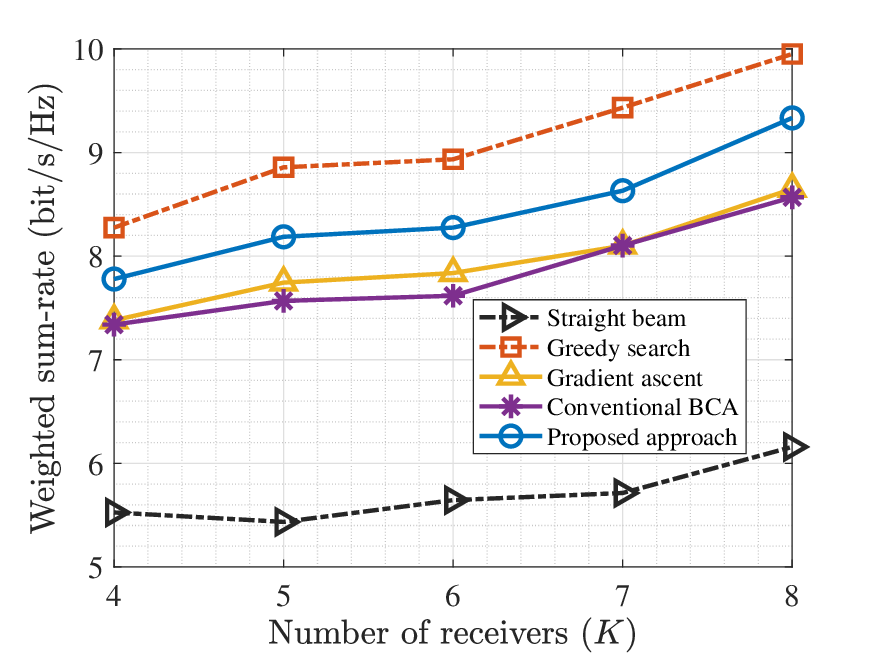}
\vspace{-1em}
\caption{ Weighted sum-rates of different methods vs. number $K$ of receivers.}
\vspace{-2em}
\label{fig:wsr_vs_user}
\end{figure}

\begin{figure}[t]
\centering
\includegraphics[width=6.8cm]{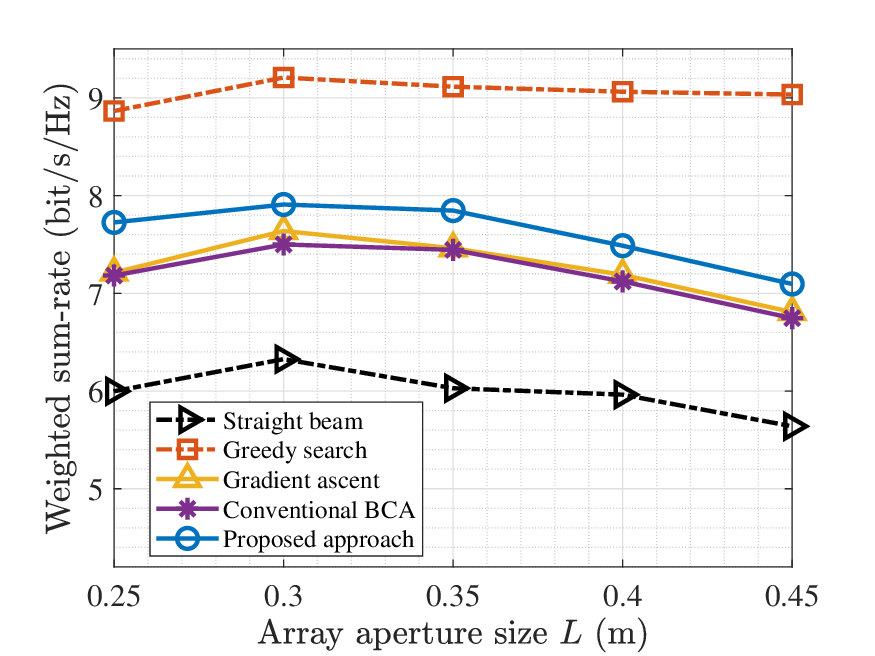}
\vspace{-1em}
\caption{ Weighted sum-rates of different methods vs. array aperture size $L$.}
\vspace{-1.14em}
\label{fig:wsr_vs_L}
\end{figure}


\begin{figure}[t]
\centering
\includegraphics[width=6.8cm]{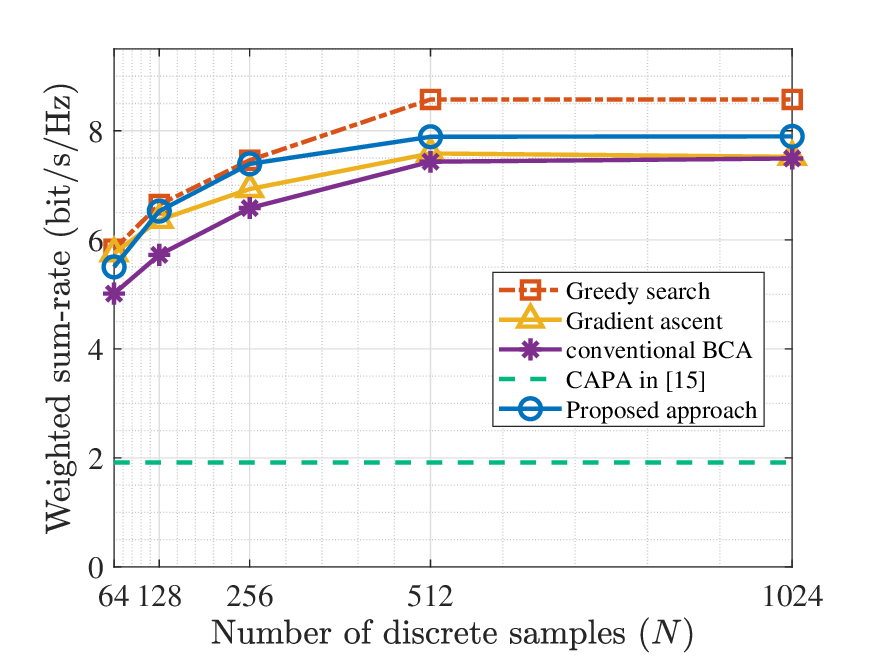}
\vspace{-1em}
\caption{ Weighted sum-rates of different methods vs. number $N$ of discrete samples}
\vspace{-2em}
\label{fig:capa}
\end{figure}

\begin{figure}[t]
\centering
\includegraphics[width=6.6cm]{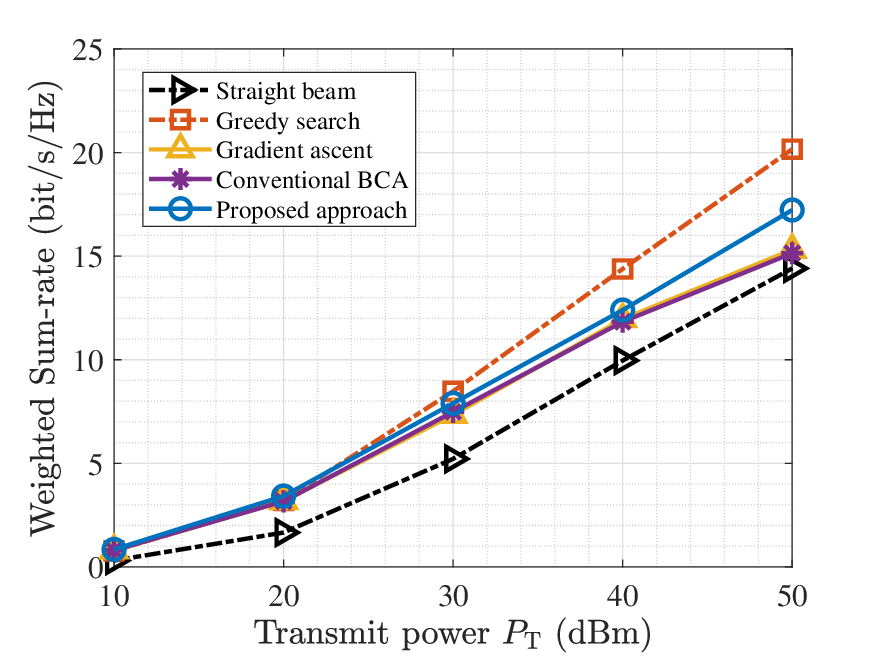}
\vspace{-1em}
\caption{ Weighted sum-rates of different methods vs. transmit power.}
\vspace{-2em}
\label{fig:wsr_vs_power}
\end{figure}


Fig.~\ref{fig:wsr_vs_user} shows the weighted sum-rate of different methods as the number \(K\) of receivers increases. In particular, when $K=8$, we observe that the proposed approach improves the weighted sum-rate by up to $58\%$ compared to the straight beam method due to the capability of curved beams to mitigate blockage-induced attenuation. From Fig.~\ref{fig:wsr_vs_user}, we also observe that the proposed approach improves the weighted sum-rate by up to \(10\%\) and \(7\%\) compared to the conventional BCA and gradient ascent methods. This is because the proposed approach leverages the local descent from previous iterations and introduces a proximal regularization term, thus enabling stable updates and reducing the possibility of being trapped in poor local optimal solutions caused by the strong coupling among optimization variables.

Fig.~\ref{fig:wsr_vs_L} shows how the weighted sum-rate changes as the array aperture size \(L\) varies. From this figure, we observe that the weighted sum-rate of all methods first increases and then decreases as \(L\) increases. This is because increasing \(L\) initially enlarges the effective aperture, which improves the spatial resolution and beamforming gain, enabling curved beams to better concentrate energy along the desired curved propagation paths. However, for a fixed number \(N\) of discrete samples, increasing \(L\) also enlarges the sampling interval \(\Delta \ell=L/N\). When \(L\ge N\times(\lambda/2)=0.256\) m, i.e., \(\Delta \ell\ge\lambda/2\), the discrete aperture no longer satisfies the half-wavelength sampling condition, which may introduce spatial aliasing and undesired sidelobe leakage. Thus, the weighted sum-rates start to decrease.

Fig.~\ref{fig:capa} compares the weighted sum-rate achieved by different curved beam optimization. Here, the CAPA method proposed in \cite{C-MIMO} uses the same equivalent channel as other curved beam-based methods, but directly designs the continuous aperture beamforming without optimizing the individual curved beam parameters. Thus, we can see that the weighted sum-rate of CAPA is independent of \(N\) and appears as a horizontal line. From this figure, we also see that the weighted sum-rates resulting from the four curved beam-based methods increase as \(N\) increases and gradually saturate. This is because a larger \(N\) provides more accurate approximation of discrete array. When \(N\) becomes sufficiently large (e.g., $N\ge 512$), the approximation error becomes negligible, so the weighted sum-rate saturates. Meanwhile, Fig.~\ref{fig:capa} also shows that the proposed approach outperforms the CAPA method even when \(N=64\) samples are used. This is because the proposed method optimizes the curved beam parameters, thus leveraging the additional degrees of freedom (DoF) introduced by curved beams. In contrast, the CAPA method optimizes the continuous aperture beamforming as a whole without exploiting the structure of curved beam parameters, which limits its performance in a blockage environment.

Fig.~\ref{fig:wsr_vs_power} shows how the weighted sum-rate changes as the transmit power $P_{\mathrm{T}}$ varies. From this figure, we observe when \(P_{\mathrm{T}}=30\) dBm, the proposed approach improves the weighted sum-rate by up to \(60\%\) and \(12\%\) compared to the straight beam method and conventional BCA methods. The \(60\%\) gain stems from the fact that the proposed approach can exploit the additional DoF offered by the curved beams. The \(12\%\) gain is because the proposed approach can utilize the calculated surrogate-construction point to escape from local optimal solutions, thus achieving a higher weighted sum-rate.

\section{Conclusion}\label{sec:conclusion}
In this paper, we have studied a curved beam
enabled wireless communication system that the transmitter can adjust the curved beam parameters to generate curved beams. To characterize the curved beam generation, we have introduced a novel curved beam model and channel model of the blockage environments. Based on the model, we have formulated a weighted sum-rate maximization problem under the transmit power budget and physical constraints of curved beam related parameters. To solve this problem, we have 
changed the continuous aperture into finite summations via discretely sampling, and analyzed the performance gap between the ideal continuous aperture design and its practical discrete aperture approximation. Based on the above discrete reformulation, we have developed an iterative algorithm to optimize the curved beam related parameters. Simulation results demonstrate that the proposed method can effectively improve the weighted sum-rate and enable reliable blockage-aware communications in wireless systems.

\section*{Appendix:\\ Proof of Theorem~\ref{theorem:gap}}
To prove \textbf{Theorem~\ref{theorem:gap}}, we first convert the  weighted sum-rate gap into an SINR gap by leveraging the triangle inequality and the \textit{Lipschitz} continuity of $\log_2(1+x)$ function on $\mathbb{R}_{++}$, which is given as follows: 
    \begin{align}
        \Big|\textstyle\sum_{k=1}^K \omega_k R_k- \sum_{k=1}^K \omega_k \tilde R_k   \Big|\stackrel{(a)}{\leq} & \textstyle\sum_{k=1}^K\omega_k| R_k-\tilde R_k|\\
        \stackrel{(b)}{\leq}&\textstyle\sum_{k=1}^K\frac{\omega_k}{\ln2}|\gamma_k-\tilde{\gamma}_k|,\label{eq:rate_gap}
    \end{align}
    where (a) follows the triangle inequality $|\sum_{k}x_k|\leq\sum_{k}|x_k|$ such that $|\sum_k \omega_k(R_k-\tilde R_k)|\leq\sum_k\omega_k|R_k-\tilde R_k|$, and (b) stems from the fact that the logarithmic function $\log_2(1+x)$ is $1/\ln2$--\textit{Lipschitz} on $\mathbb{R}_{++}$ and hence we have $\log_2(1+\gamma_k)\leq \gamma_k/\ln2$ and $\gamma_k\in \mathbb{R}_{++}$.

    Based on \eqref{eq:rate_gap}, next, we analyze the gap between the SINR $\gamma_k$ and $\tilde{\gamma}_k,k\in\cK$. To obtain compact forms of $\gamma_k$ and $\tilde{\gamma}_k$, we first define the following functions
    \begin{align}
    &F_{k,i}(\ell)\triangleq J_i(\ell)\tilde{h}(
        \bm p_k,\bm \ell).\label{eq:F} \\
        &g_{k,i}\triangleq\textstyle\int_{\mathcal{L}} F_{k,i}(\ell)\,d\ell, \quad \tilde g_{k,i}\triangleq\sum_{n=1}^N F_{k,i}(\ell_n)\Delta_\ell. \label{eq:g}\\
        &D_{k}\triangleq\textstyle\sum_{i\neq k}|g_{k,i}|^2+\sigma_k^2, \quad\tilde{D}_{k}\triangleq\textstyle\sum_{i\neq k}|\tilde{g}_{k,i}|^2+\sigma_k^2.\label{eq:D}
    \end{align}
    By substituting \eqref{eq:F}--\eqref{eq:D} into $|\gamma_k-\tilde{\gamma}_k|,k\in\cK$, we have
    \begin{align}\label{eq:gap_gamma}
    |\gamma_k-\tilde{\gamma}_k|=&\,\bigg|
\frac{|g_{k,k}|^2}{D_{k}}
-
\frac{|\tilde g_{k,k}|^2}{\tilde{D}_{k}}
\bigg| \nonumber\\
=&\, \bigg| \frac{|g_{k,k}|^2\big(\tilde{D}_{k}-D_{k}\big)+\big( |g_{k,k}|^2-|\tilde{g}_{k,k}|^2 \big)D_k}{D_{k}\tilde{D}_{k}}  \bigg|.
    \end{align}
    Then, by the triangle inequality and the definitions of $D_k,\tilde{D}_k$ in \eqref{eq:D}, we obtain the the upper bound of \eqref{eq:gap_gamma} as follows:
    \begin{align}\label{eq:sinr_gap}
        &|\gamma_k-\tilde{\gamma}_k|\nonumber\\\stackrel{(a)}{\le}&
\frac{
|g_{k,k}|^2
\left|
\sum_{i\neq k}\big(|\tilde g_{k,i}|^2-|g_{k,i}|^2\big)
\right|
}{
D_k\tilde{D}_k
}
+
\frac{
\big||g_{k,k}|^2-|\tilde g_{k,k}|^2\big|
}{D_k
} \nonumber\\
\stackrel{(b)}{\le}&
\frac{|g_{k,k}|^2}{\sigma_k^4}
\bigg|
\sum_{i\neq k}\big(|\tilde g_{k,i}|^2-|g_{k,i}|^2\big)
\bigg|
+
\frac{1}{\sigma_k^2}\bigg|\big|g_{k,k}|^2-|\tilde g_{k,k}\big|^2\bigg| \nonumber\\
\stackrel{(c)}{\le}&
\frac{2|g_{k,k}|}{\sigma_k^2}\,|\tilde g_{k,k}-g_{k,k}|
+\frac{2|g_{k,k}|^2}{\sigma_k^4}\sum_{i\neq k}|g_{k,i}|\,|\tilde g_{k,i}-g_{k,i}|\nonumber\\& +
\bigg(
\frac{1}{\sigma_k^2}|\tilde g_{k,k}-g_{k,k}|^2
+\frac{|g_{k,k}|^2}{\sigma_k^4}\sum_{i\neq k}|\tilde g_{k,i}-g_{k,i}|^2
\bigg),
    \end{align}
where (a) and (c) follow the triangle inequality, (b) follows $D_k\geq \sigma_k^2, \tilde{D}_k\geq \sigma_k^2$. From \eqref{eq:sinr_gap}, we see that the gap $|\gamma_k-\tilde{\gamma}_k|,k\in\cK,$ is upper-bounded by $|\tilde g_{k,i}-g_{k,i}|,  k,i\in\cK$. 


Next, we characterize the gap $|\tilde g_{k,i}-g_{k,i}|, \forall k,i\in\cK$. By substituting \eqref{eq:g} into $|\tilde g_{k,i}-g_{k,i}|$, we have
    \begin{align}\label{eq:11}
        \Big|\tilde g_{k,i}-g_{k,i}\Big|=&\bigg|\int_{\mathcal{L}} F_{k,i}(\ell)\,d\ell-\sum_{n=1}^N F_{k,i}(\ell_n)\Delta_\ell\bigg|.
    \end{align}
To simplify \eqref{eq:11}, we first rewrite the continuous integral term (i.e., $\int_{\mathcal{L}} F_{k,i}(\ell)\,d\ell$) into a form that can be directly compared with the discrete term (i.e., $\sum_{n} F_{k,i}(\ell_n)\Delta_\ell$). In particular, we devide the interval $\mathcal{L}=[-L/2,L/2]$ into $N$ subintervals, such as $I_n=[\ell_n-\Delta_\ell/2,\ \ell_n+\Delta_\ell/2]$.
    For any $\ell\in I_n$, by employing the \textit{second-order Taylor expansion} of $F_{k,i}$ around $\ell_n$, we have
    \begin{align}\label{eq:te_1}
       F_{k,i}(\ell)=F_{k,i}(\ell_n)&+\frac{\partial F_{k,i}(\ell_n)}{\partial \ell}(\ell-\ell_n)\nonumber\\
       &+\frac{1}{2}\times\frac{\partial^2 F_{k,i}(\zeta_{n,\ell})}{\partial \ell^2}(\ell-\ell_n)^2,
    \end{align}
    where $\zeta_{n,\ell}\in I_n$. Based on \eqref{eq:te_1}, we have
    \begin{align}\label{eq:22}
        \int_{I_n} F_{k,i}(\ell)\,d\ell= &F_{k,i}(\ell_n)\Delta_\ell\nonumber\\&+\frac{1}{2} \int_{I_n}\frac{\partial^2 F_{k,i}(\zeta_{n,\ell})}{\partial \ell^2}(\ell-\ell_n)^2d\ell,
    \end{align}
    where $\int_{I_n}(\ell-\ell_n)d\ell=0$ due to the symmetry of the interval $I_n$ with respect to its midpoint $\ell_n$. By substituting \eqref{eq:22} into \eqref{eq:11}, we have
    \begin{align}\label{eq:g_gap}
        \Big|\tilde g_{k,i}-g_{k,i}\Big|=&\frac{1}{2}\bigg|\sum_{n=1}^N\int_{I_n}\frac{\partial^2 F_{k,i}(\zeta_{n,\ell})}{\partial \ell^2}(\ell-\ell_n)^2d\ell\bigg|\nonumber\\
        \stackrel{(a)}{\leq}&\frac{1}{2}\sum_{n=1}^N\int_{I_n}\bigg|\frac{\partial^2 F_{k,i}(\zeta_{n,\ell})}{\partial \ell^2}\bigg|(\ell-\ell_n)^2d\ell\nonumber\\
        \leq& \frac{1}{2}\times\sup_{\zeta\in\mathcal{L}}\bigg|\frac{\partial^2 F_{k,i}(\zeta)}{\partial \ell^2}\bigg|\times\sum_{n=1}^N\int_{I_n}(\ell-\ell_n)^2 d\ell\nonumber\\
        =&\sup_{\zeta\in\mathcal{L}}\bigg|\frac{\partial^2 F_{k,i}(\zeta)}{\partial \ell^2}\bigg|\times\frac{L^3}{24N^2}.
    \end{align}
    where $\zeta$ is an arbitrary point in the interval $[-L/2,L/2]$, and (a) follows the triangle inequality.
   
    From \eqref{eq:g_gap}, we see that if we want to characterize the gap $|\tilde g_{k,i}-g_{k,i}|, \forall k,i\in\cK$, we must calculate the value of $\sup_{\zeta\in\mathcal{L}}\big|\partial^2 F_{k,i}(\zeta)/\partial \ell^2\big|, \forall k,i\in\cK$. Since the effective channel $\tilde{h}(\bm p_0,\bm \ell)$ in \eqref{eq:eff_channel} as well as the function $F_{k,i}$ is twice differentiable with respect to $\ell$, and that all aperture amplitudes are identical, i.e., $\alpha_k(\ell)=\alpha\geq0$, $\forall k\in\cK,\ \forall \ell\in\mathcal L$, as stated in Theorem~\ref{theorem:gap}, we have
    \begin{align}\label{eq:BB}
         \bigg|\frac{\partial^2 F_{k,i}(\ell)}{\partial \ell^2}\bigg|\stackrel{(a)}{\leq}&\Big|\frac{\partial^2 J_{i}(\ell)}{\partial \ell^2}\Big|\Big|\tilde{h}(\bm p_k, \bm\ell)\Big|+\Big|\frac{\partial^2\tilde{h}(\bm p_k, \bm\ell)}{\partial\ell^2}\Big|\Big|J_i({\ell})\Big|\nonumber\\&\qquad\quad\quad\quad+2\Big|\frac{\partial J_{i}(\ell)}{\partial \ell}\Big|\Big|\frac{\partial\tilde{h}(\bm p_k, \bm\ell)}{\partial\ell}\Big|,
    \end{align}
    where (a) follows the triangle inequality and the multiplicativity of the absolute value. Combining \eqref{eq:BB} and the definition of $\sup_{\zeta\in\mathcal{L}}\big|\partial^2 F_{k,i}(\zeta)/\partial \ell^2\big|$, we have
    \begin{align}
        \sup_{\zeta\in\mathcal{L}}\bigg|\frac{\partial^2 F_{k,i}(\zeta)}{\partial \ell^2}\bigg|=\Psi_{i}^{(2)}\Sigma_k^{(0)}+\Psi_{i}^{(0)}\Sigma_k^{(2)}+2\Psi_{i}^{(1)}\Sigma_k^{(1)},
    \end{align}
    where 
    \begin{align}
        &\Psi_i^{(q)}\triangleq\sup_{\ell\in\mathcal{L}}  \big|\partial^q J_{i}(\ell)/\partial \ell^q\big|,\quad q=0,1,2,\\
        &\Sigma_i^{(q)}\triangleq\sup_{\ell\in\mathcal{L}}  \big|\partial^q \tilde{h}(\bm p_k, \bm\ell)/\partial \ell^q\big|,\quad q=0,1,2.
    \end{align}
    To calculate the upper bound of $\sup_{\zeta\in\mathcal{L}}\big|\partial^2 F_{k,i}(\zeta)/\partial \ell^2\big|$, we derive the upper bound of $\Psi_i^{(q)}$ and $\Sigma_k^{(q)}$, respectively.
    
   \textbf{1) For $\Psi_i^{(q)},q=0,1,2$:} From \eqref{eq:beam_steering}--\eqref{eq:get_J_k}, we see that $J_i(\ell),i\in\cK$ is smooth with respect to $\ell$ over the interval $\mathcal{L}$. Meanwhile, $\ell$ is defined in the interval $\mathcal{L}=[-L/2,L/2]$ and $|e^{j\beta_k(\ell)}|=1$. Hence, we have 
    \begin{align}
        &\Psi_i^{(0)}=\alpha\triangleq \tilde{\Psi}^{(0)}_i,\label{eq:psi0}\\
        &\Psi_i^{(1)}\stackrel{(a)}\le \alpha\bigg(\frac{2\pi}{\lambda}|\sin\theta_i|+\frac{2\pi L}{\lambda f_i}+\frac{(2\pi B_i)^3L^2}{4}\bigg)\triangleq \tilde{\Psi}^{(1)}_i,\label{eq:psi1}\\
        &\Psi_i^{(2)}\stackrel{(b)}\le \alpha
    \Bigg[
    \frac{4\pi}{\lambda f_i}
    +(2\pi B_i)^3L+
    \bigg(
    \frac{2\pi}{\lambda}|\sin\theta_i|\nonumber\\
    &\qquad\qquad\qquad\qquad
    +\frac{2\pi L}{\lambda f_i}
    +\frac{(2\pi B_i)^3L^2}{4}
    \bigg)^2
    \Bigg]\triangleq \tilde{\Psi}^{(2)}_i.\label{eq:psi2}
    \end{align}
    where (a) and (b) stem from the fact that $\max_{\ell\in\mathcal{L}}|\ell|=L/2$.

    \textbf{2) For $\Sigma_k^{(q)},q=0,1,2$:} From \eqref{eq:blockage_function}, we see that the blockage indicator $\Gamma(\bm r_m)$ with $|\Gamma(\bm r_m)|<1$ is independent of $\ell$. Thus, we have
    \begin{align}\label{eq:norm_G_bound}
        &\quad\bigg|\frac{\partial^q\tilde{h}(\bm p_k, \bm\ell)}{\partial\ell^q}\bigg|\nonumber\\=&\sqrt{\frac{{\epsilon_k}}{\int_{\mathcal{L}}\big|\hat{h}(\bm p_k,\bm \ell)\big|^2 d\ell}}\times\bigg|\frac{\partial^q\hat{h}(\bm p_k, \bm\ell)}{\partial\ell^q}\bigg|\nonumber\\\leq &\sqrt{\frac{{\epsilon_k}}{\int_{\mathcal{L}}\big|\hat{h}(\bm p_k,\bm \ell)\big|^2 d\ell}}\times\bigg|\frac{\partial^qh(\bm r_1,\bm p_0)}{\partial\ell^q}\bigg|\times\prod_{m=2}^{M}\Big|h(\bm r_m,\bm r_{m-1})\Big|.
    \end{align}
    Next, we analyze $|\partial^qh(\bm r_1,\bm p_0)/\partial\ell^q|$ and $|h(\bm r_m,\bm r_{m-1})|$, respectively.

    \begin{itemize}
        \item \textit{For $|\partial^qh(\bm r_1,\bm p_0)/\partial\ell^q|$:} We first introduce the auxiliary variables $t$ and $t_{+}$ to replace the $\ell$-dependent terms, i.e.,
        \begin{align}
           t\triangleq y_1-\ell, \qquad t_+\triangleq \textstyle\min_{\ell} |t|=(|y_1|-L/2)^{+}. \label{eq:t}
        \end{align}
        By substituting \eqref{eq:t} into \eqref{eq:eff_channel}, we obtain the upper bound of $|\partial^qh(\bm r_1,\bm p_0)/\partial\ell^q|,q=0,1,2$ as specified in \eqref{eq:d0_h}--\eqref{eq:d2_h},\begin{figure*}[t]
            \begin{align}\label{eq:d0_h}
            &|h(\bm r_1,\bm p_0)|=\bigg|\frac{1}{\Delta_x^2+t^2}-\frac{j2\pi}{\lambda}\bigg|\times\bigg|\frac{\Delta_x}{2\pi (\Delta_x^2+t^2)}\bigg|\stackrel{(a)}{\leq} \frac{c_0}{\Delta_x^2+t^2}\stackrel{(b)}{\leq} \frac{c_0}{\Delta_x^2+t^2_{+}},\\
            \bigg|\frac{\partial}{\partial\ell}h(\bm r_1,\bm p_0)\bigg|&=\frac{t\Delta_x}{2\pi} \times \bigg|\frac{4\pi^2}{\lambda^2(\Delta_x^2+t^2)^{3/2}}+\frac{j6\pi}{\lambda(\Delta_x^2+t^2)^{2} }-\frac{6\pi}{(\Delta_x^2+t^2)^{5/2}}\bigg|\stackrel{(c)}{\leq}\frac{c_1}{\Delta_x^2+t^2}\stackrel{(d)}{\leq}\frac{c_1}{\Delta_x^2+t^2_{+}}, \label{eq:d1_h}\\
            \bigg|\frac{\partial^2}{\partial\ell^2}h(\bm r_1,\bm p_0)\bigg|&=\frac{\Delta_x}{2\pi}\times\bigg[t^2\bigg(\frac{j8\pi^3}{\lambda^3(\Delta_x^2+t^2)^{2}}-\frac{20\pi^2}{\lambda^2(\Delta_x^2+t^2)^{5/2}}-\frac{j24\pi}{\lambda(\Delta_x^2+t^2)^{3}}+\frac{12}{(\Delta_x^2+t^2)^{7/2}}\bigg)\nonumber\\&\qquad\qquad\qquad+\Delta_x^2\bigg(\frac{4\pi^2}{\lambda^2(\Delta_x^2+t^2)^{5/2}}+\frac{j3\pi}{\lambda(\Delta_x^2+t^2)^{3}}-\frac{3}{(\Delta_x^2+t^2)^{7/2}} \bigg)  \bigg]\stackrel{(e)}{\leq}\frac{c_2}{\Delta_x^2+t^2}\stackrel{(f)}{\leq}\frac{c_2}{\Delta_x^2+t^2_{+}}.\label{eq:d2_h}
        \end{align}
        \hrule \vspace{-1.5em}
        \end{figure*} where (a), (c), (e) follow the triangle inequality and $\Delta_x^2+t^2\geq \Delta^2_x$, (b), (d), (f) stem from the fact that $t_+^2 \leq t^2$ in \eqref{eq:t}, and
        \begin{align}
            &c_0=\frac{1}{2\pi}+\frac{\Delta_x}{\lambda},\label{eq:c_0}\\
            &c_1=\frac{\Delta_x}{2\pi}\times\bigg(\frac{4\pi^2}{\lambda^2}+\frac{6\pi}{\lambda\Delta_x}+\frac{3}{\Delta_x^2}\bigg),\\
            &c_2=\frac{\Delta_x}{2\pi}\times\bigg(\frac{8\pi^3}{\lambda^3}+\frac{24\pi^2}{\lambda^2\Delta_x}+\frac{30\pi}{\lambda\Delta_x^2}+\frac{15}{\Delta_x^3}\bigg). 
        \end{align}
       
        \item \textit{For $|h(\bm r_m,\bm r_{m-1})|,m\geq 2$:} Similar to \eqref{eq:d0_h} and \eqref{eq:c_0}, we obtain the upper bound of $|h(\bm r_m,\bm r_{m-1})|,m\geq 2$ as follows:
        \begin{align}\label{eq:h_bound}
            |h(\bm r_m,\bm r_{m-1})|\stackrel{(a)}{\leq}& \frac{c_0}{\Delta_x^2+(y_m-y_{m-1})^2}, \quad m\geq2,
        \end{align}
        where (a) follows the triangle inequality and $\Delta_x^2+(y_m-y_{m-1})^2\geq \Delta_x^2$.

    \end{itemize}
    Hence, by substituting \eqref{eq:d0_h}--\eqref{eq:d2_h} and \eqref{eq:h_bound}
    into \eqref{eq:norm_G_bound}, we have
    \begin{align}\label{eq:bound_h}
        \Sigma_k^{(q)}\leq\, \sqrt{\frac{{\epsilon_k}}{\int_{\mathcal{L}}\big|\hat{h}(\bm p_k,\bm \ell)\big|^2 d\ell}}\bigg(\frac{\pi}{\Delta_x}\bigg)^Mc_q c_0^{M-1}\triangleq \tilde{\Sigma}^{(q)}_k.
    \end{align}
    Combining the upper bounds of $\Psi_i^{(q)}, \forall i$ in \eqref{eq:psi0}--\eqref{eq:psi2} and $\Sigma_k^{(q)},\forall k$ in \eqref{eq:bound_h}, we obtain the upper bound of $\sup_{\zeta\in\mathcal{L}}|\partial^2 F_{k,i}(\zeta)/\partial \ell^2|, \forall k,i$, i.e.,
    \begin{align}\label{eq:M}
        &\textstyle\sup_{\zeta\in\mathcal{L}}\ \Big|\partial^2 F_{k,i}(\zeta)/\partial \ell^2\Big|\times L^3/24\nonumber\\\leq& \Big(\tilde\Psi_{i}^{(2)}\tilde\Sigma_k^{(0)}+\tilde\Psi_{i}^{(0)}\tilde\Sigma_k^{(2)}+2\tilde\Psi_{i}^{(1)}\tilde\Sigma_k^{(1)}\Big)\times L^3/24\triangleq  D_{k,i},
    \end{align}

Finally, by substituting \eqref{eq:sinr_gap}, \eqref{eq:g_gap}, \eqref{eq:M} into \eqref{eq:rate_gap}, we have
\begin{align}
    &\bigg|\sum_{k=1}^K \omega_k R_k- \sum_{k=1}^K \omega_k \tilde R_k   \bigg|\nonumber\\\leq\, &\sum_{k=1}^K\bigg[\frac{\omega_k}{N^2\ln{2}}\bigg(\frac{2|g_{k,k}|}{\sigma_k^2}D_{k,k}+\frac{2|g_{k,k}|^2}{\sigma_k^4}\sum_{i\neq k}|g_{k,i}|D_{k,i}\bigg)\nonumber\\
    &\qquad\qquad\ +\frac{\omega_k}{N^4\ln{2}}\bigg(\frac{D_{k,k}^2}{\sigma_k^2}+\frac{|g_{k,k}|^2}{\sigma_k^4}\sum_{i\neq k}D^2_{k,i}\bigg)\bigg].
\end{align}
This completes the proof.

\vspace{-0.5em}

\bibliographystyle{IEEEtran}     
\bibliography{IEEEabrv,Strings}

\end{document}